\documentclass{llncs}

\usepackage[utf8]{inputenc}
\usepackage{amsmath}
\usepackage{amsfonts}
\usepackage{amssymb}
\usepackage{mathtools}
\usepackage{microtype}

\usepackage{pgfplots}
\pgfplotsset{compat=1.18} 

\usepackage[colorlinks=true, linkcolor=blue, citecolor=blue, urlcolor=blue]{hyperref}

\newenvironment{proofsketch}{\par\addvspace{1ex}\noindent\textit{Proof sketch.}\enspace}{\qed\medskip}
\spnewtheorem{observation}[theorem]{Observation}{\bfseries}{\itshape}

\allowdisplaybreaks[1]

\title{Trading Proportionality for Strategic Robustness in Multi-Winner Approval Voting}
\author{Dael Sinay\inst{1} \and Rica Gonen\inst{1}}
\authorrunning{D. Sinay and R. Gonen}
\institute{Dept. of Mathematics and Computer Science, The Open University of Israel \\
	\email{dael.algo@gmail.com, ricagonen@gmail.com}}

\date{\small \today}

\begin{document}
	
	\maketitle
	
	\begin{abstract}
		Classical strategyproofness assumes a manipulator either knows how everyone else votes or is willing to gamble as if they did. Real voters rarely do. The recently introduced RAT-degree measures how many other participants' reports an agent must actually observe before a manipulation becomes strictly safe, interpolating between full truthfulness and immunity to blind manipulation. While previously explored in auctions and single-winner settings, we bring this measure to multi-winner elections. We apply it to Approval-Based Committee (ABC) rules under free-riding: a voter drops approved candidates from their truthful ballot to concentrate weight on marginal ones.
		
		We first analyze Proportional Approval Voting (PAV). Knowledge of $\lceil n/k \rceil$ ballots already enables a safe and strictly profitable drop, whereas knowledge of at most $\lfloor n/(k+1) \rfloor - 1$ ballots leaves the rule completely immune; an explicit instance shows the latter bound cannot be raised in general. Since a manipulator informed about roughly a $1/k$ fraction of the electorate therefore suffices, we ask how much proportionality must be surrendered to buy strategic robustness. We introduce $d$-RPAV, a parameterized family of Thiele rules with weights $d/(j+d-1)$ that recovers PAV at $d = 1$ and approaches Approval Voting (AV) as $d$ grows. We prove that $d$-RPAV satisfies $\alpha$-Justified Representation ($\alpha$-JR) for $\alpha = d$, and is immune to safe free-riding given up to $\lfloor dn/(k+2d-1) \rfloor - 1$ known ballots, yielding a clean and tunable trade-off between proportional representation and strategic robustness.
	\end{abstract}	
	\newpage
	
	\section{Introduction}
	\label{sec:intro}
	
	Committees, boards, shortlists, and parliaments are chosen by multi-winner
	elections, and the rules that choose them are expected to do two things at once.
	They should be \emph{proportional}: a cohesive group holding some fraction of
	the electorate should control a corresponding fraction of the $k$ elected  seats. They
	should also be \emph{strategyproof}: no voter should gain by misreporting. In
	Approval-Based Committee (ABC) elections, where each voter simply reports the
	set of candidates they find acceptable, these two demands are known to be
	irreconcilable. Peters proved that no ABC rule satisfies even minimal
	proportionality together with strategyproofness~\cite{peters2018}, and the
	tension is visible in the two best-known rules at the extremes. Multi-winner
	Approval Voting (AV) elects the $k$ most-approved candidates; it is
	strategyproof and leaves cohesive minorities entirely unrepresented.
	Proportional Approval Voting (PAV), which scores a voter with $x$ approved
	candidates in the elected committee by the harmonic sum $H(x) = 1 + \tfrac{1}{2} + \dots + \tfrac{1}{x}$,
	satisfies strong proportionality axioms such as Extended Justified
	Representation (EJR)~\cite{Aziz2017} and is manipulable.
	
	Impossibility results of this kind, however, and the classical definition of
	strategyproofness that underlies them, grant the manipulator a great deal. A
	voter is counted as having an incentive to lie if \emph{some} profile of the
	other ballots makes lying profitable. This implicitly assumes that the voter
	either knows what everyone else will report, or is willing to act as though
	they do. Actual voters are in a much weaker position. They face profound
	uncertainty about the rest of the electorate, and most deviations from the
	truthful ballot are gambles: profitable against some possible completions of the
	profile, harmful against others. A risk-averse voter will deviate only when the
	deviation is \emph{safe}, never strictly worse than truth-telling under any
	realization of the unknown ballots, and strictly better under at least one.
	This behavior motivates the standard of \emph{Risk-Avoiding Truthfulness}
	(RAT)~\cite{bu2023existence}, which requires immunity only to such safe
	deviations, shifting the requirement from classical strategyproofness to
	protection against risk-averse manipulators.
	
	To explore the spectrum between these two requirements, the \emph{RAT-degree}~\cite{hartman2025}
	turns the binary question ``is this rule manipulable?'' into a quantitative
	one: \emph{how many other ballots must a voter observe before some manipulation
		becomes safe?} A rule with RAT-degree $f$ is immune to safe manipulation by any
	voter who knows fewer than $f$ of the other ballots. The measure interpolates
	between classical strategyproofness (degree $n$, where even an omniscient
	manipulator has no safe deviation) and bare RAT (degree $0$). Conceptually, it is an
	informational measure rather than a computational one: it asks what a
	manipulator must \emph{know}, not what they must compute. The RAT-degree was originally studied and determined in auctions, indivisible-goods allocation, cake-cutting, and
	two-sided matching, but multi-winner elections, the setting in which the
	proportionality/strategyproofness conflict is sharpest, have not been
	examined. This paper closes that gap.
	
	We study the manipulation that matters most for proportional ABC rules:
	\emph{free-riding}, or \emph{dropping candidates}~\cite{hylland1992proportional,schulze2004free}.
	A voter omits from their ballot a candidate whose election already looks
	secure, so that their voting weight is not diluted by diminishing marginal
	returns and instead pushes a marginal candidate over the line. Free-riding is
	exactly the behavior that proportional rules invite and that AV does not, and
	it is the mechanism behind Peters' impossibility. Combining it with the
	risk-avoiding lens, and with the standard superset criterion for when a
	committee is an improvement, yields the property we analyze: \emph{Superset
		$f$-Risk-Avoiding Strategy-Proofness under Dropping Candidates}
	(S-$f$-RASProof-DC), which holds when no voter who knows $f$ other ballots has a
	safe and strictly profitable drop.
	
	\subsection*{Our Contributions}
	
	\paragraph{A knowledge threshold for PAV.}
	We locate the point at which free-riding under PAV becomes safe. A voter who
	knows the ballots of $\lceil n/k \rceil$ other voters can drop a candidate
	safely and profitably (Theorem~\ref{thm:pav-vul}): a quota of that size
	already guarantees the dropped candidate's election, so the drop carries no
	risk, while the recovered marginal weight can decide a knife-edge race for the
	last seat. Conversely, PAV is completely immune to safe drops against any voter
	who knows at most $\lfloor n/(k+1) \rfloor - 1$ ballots
	(Theorem~\ref{thm:pav-proofness}). An explicit six-voter instance shows the
	immunity bound cannot be raised in general (Example~\ref{ex:tightness}), so the
	remaining gap is narrow and cannot be closed by a better argument alone. The
	headline reading is uncomfortable: the required knowledge scales as a $1/k$
	fraction of the electorate, so the larger the committee, the \emph{less} a
	voter needs to know. Electing a ten-seat body makes PAV safely manipulable by
	anyone with a reliable picture of roughly a tenth of the voters, such as a polled
	district, a union branch, or an organized bloc.
	
	\paragraph{A tunable family that buys robustness with proportionality.}
	Since PAV's immunity is thin and full strategyproofness is unreachable, the
	natural question is what a rule can gain by conceding proportionality by a
	controlled amount. Consider relaxing the guarantee from ``$\lceil n/k \rceil$
	unanimous supporters elect a candidate'' to ``$\lceil 2n/k \rceil$ supporters
	do'', then $\lceil 3n/k \rceil$, and so on. We show this intuition has a clean
	analytic counterpart. We introduce \emph{$d$-Regressive Proportional Approval
		Voting} ($d$-RPAV), the Thiele family whose $j$-th approved candidate
	contributes $\frac{d}{j+d-1}$. At $d = 1$ the rule is exactly PAV; as $d$ grows
	the marginal returns flatten toward those of AV, tracing the \emph{regressive}
	half of the Thiele spectrum~\cite{Jaworski2022PhragmnRF,lackner2023multi}. We prove two
	matching facts about this family:
	
	\begin{itemize}
		\item $d$-RPAV satisfies $\alpha$-Justified Representation~\cite{becker2026}
		for $\alpha = d$, and hence $\tfrac{1}{d}$-proportionality
		(Theorem~\ref{thm:drpav-alpha-jr}), by a pigeonhole-and-swap argument on
		marginal contributions;
		\item $d$-RPAV is immune to safe free-riding against any voter knowing at
		most $\left\lfloor \frac{dn}{k+2d-1} \right\rfloor - 1$ ballots
		(Theorem~\ref{thm:drpav-proofness}), which recovers the PAV bound at
		$d = 1$.
	\end{itemize}
	
	The guarantee is increasing in $d$ for every $k \ge 2$, so the parameter is a
	genuine dial: one turn weakens the representation quota by a known factor and
	strengthens the informational barrier to manipulation by a known amount. As $d \to \infty$, this barrier approaches $n/2$ known ballots. We conjecture that this limit represents a fundamental ceiling for any \emph{reasonable} multi-winner rule other than AV, leaving the formalization of this intuition to future work.
	
	\paragraph{Technique.}
	The lower bound is the technical core. Given an arbitrary voter, an arbitrary
	pair of truthful and dropped ballots, and an arbitrary set of $f$ observed
	ballots, we must exhibit a single completion of the profile under which the
	drop backfires. We construct two committees differing in one seat, where one
	contains a dropped candidate $c'$ and the other contains an outsider $c''$,
	and tune how many unknown voters approve each so that $c'$ sits precisely on
	the seam: it wins under the truthful ballot and loses under the manipulated
	one. Making that seam land at an integer number of voters is what the bound
	$\left\lfloor \frac{dn}{k+2d-1} \right\rfloor - 1$ encodes. The construction
	must then survive comparison against every other committee, which we handle
	with a swap lemma showing that any committee reached by exchanging a
	constructed candidate for an unused one strictly loses score.
	
	Because $d$-RPAV stays inside the Thiele family, the substantial literature on
	computing and approximating Thiele winners applies to it unchanged; the
	questions we study here are informational and strategic, and are orthogonal to
	winner-determination complexity.
	
	\subsection*{Organization}
	
	Section~\ref{sec:prelim} fixes notation, recalls Thiele rules, defines
	$d$-RPAV, and develops the hierarchy of strategyproofness and proportionality
	definitions used throughout, culminating in S-$f$-RASProof-DC and the associated degree.
	Section~\ref{sec:pav-rule} presents the vulnerability and immunity
	thresholds for PAV and the tightness instance. Section~\ref{sec:drpav-rule} establishes the representation and robustness guarantees for $d$-RPAV. Full proofs are deferred to the
	appendix.

	\section{Preliminaries}
	\label{sec:prelim}
	
	In this section, we formally introduce the framework of approval-based committee selection and define the core concepts used throughout the paper.
	
	\subsection{Problem Definition}
	
	Let $C$ be a finite set of $m$ candidates, and let $N = \{1, \dots, n\}$ be a finite set of $n$ voters. 
	
	An \emph{approval ballot} is a proper subset $P_i$ of $C$ such that $\emptyset \neq P_i \subsetneq C$. We denote the set of all possible approval ballots by $\mathcal{B}$. An \emph{approval profile} is a function that assigns to each voter an approval ballot. For brevity, we write a profile as a tuple, such that $P = (P_1, \dots, P_n)$.
	
	Let $k$ be a fixed integer satisfying $1 \le k < m$. A \emph{committee} is a subset of $C$ with cardinality $k$. We denote the set of all possible committees by $C_k$. An (approval-based) \emph{committee rule} is a function $\mathcal{R}: \mathcal{B}^n \to C_k$, which assigns a unique winning committee to each approval profile.

	\subsection{ Thiele Rules}	
	
	A prominent class of approval-based committee rules is the family of \emph{Thiele rules}. A Thiele rule is defined by a non-decreasing function $\mathit{weight}: \mathbb{N} \cup \{0\} \to \mathbb{R}_{\ge 0}$, where $\mathit{weight}(x)$ represents the score a voter assigns to a committee that contains exactly $x$ of their approved candidates. We assume that $\mathit{weight}(0) = 0$. For a given profile $P$, the total score of a committee $W \in C_k$ is calculated directly as:
	
	\[
	\text{score}_{\mathit{Thiele}}(W) = \sum_{i \in N} \mathit{weight}(|P_i \cap W|)
	\]
	
	The rule then selects a committee $W \in C_k$ that maximizes this total score. As we will see, different choices of the $\mathit{weight}$ function yield rules with vastly different properties regarding minority representation.
	
	The most straightforward Thiele rule is Approval Voting (AV)~\cite{kilgour2010approval,lackner2023multi}. In AV, every elected candidate that a voter approves contributes exactly $1$ to their satisfaction. Thus, the weight function is simply $\mathit{weight}(x) = x$. In practice, this is equivalent to independently counting the total number of approvals each candidate receives and selecting the $k$ candidates with the highest totals.

	\begin{definition}[Approval Voting (AV)]
		The AV rule is the committee rule that selects the committee $W \in \mathcal{C}_k$ that maximizes the AV score:
		\[
		\mathit{score}_{AV}(W) = \sum_{i \in N} \left| P_i \cap W\right|
		\]
		Ties between committees with the same score are broken by a lexicographic order.
	\end{definition}

	AV has several highly desirable properties: it is computationally efficient, and it is truthful, meaning voters never have an incentive to misreport their preferences. However, its primary flaw lies in its failure to guarantee fair representation for minorities. Because AV overwhelmingly favors large majorities, cohesive minority groups can easily be left entirely unrepresented, a phenomenon often referred to as the tyranny of the majority. To illustrate this lack of representation, consider the following example.

	\
	\begin{example} \label{ex:representation_failure}
		Let $N = \{1, \dots, 10\}$ be a set of $n=10$ voters, and let $C = \{c_1, \dots, c_9\}$ be a set of candidates. We wish to elect a committee of size $k=5$. The approval profile $P$ is defined as follows:
		\begin{itemize}
			\item Voter 1 approves $\{c_1, c_3\}$.
			\item Voter 2 approves $\{c_2, c_3\}$.
			\item Voters 3 and 4 approve $\{c_4\}$.
			\item Voters 5 through 10 (a total of 6 voters) approve $\{c_5, c_6, c_7, c_8, c_9\}$.
		\end{itemize}
	\end{example}

	Let us examine the outcome of this election under the AV rule. We independently sum the approvals for each candidate: candidates $c_5, c_6, c_7, c_8$, and $c_9$ receive 6 approvals each; $c_3$ and $c_4$ receive 2 approvals each; and $c_1$ and $c_2$ receive 1 approval each. To maximize the AV score, the rule selects the $k$ candidates with the highest totals, resulting in the winning committee $W_{AV} = \{c_5, c_6, c_7, c_8, c_9\}$. Notice that the majority group of six voters completely dictates the composition of the committee. The remaining four voters, constituting $40\%$ of the electorate, receive absolutely no representation.
	
	To address the severe under-representation seen in AV, one might consider the opposite extreme among Thiele rules: the Chamberlin-Courant (CC) rule. Under CC, a voter's satisfaction is entirely binary; they only care about having at least one approved candidate on the committee, and derive no additional value from having more. Thus, the weight function is $\mathit{weight}(0) = 0$ and $\mathit{weight}(x) = 1$ for all $x \ge 1$.

	\begin{definition}[Chamberlin-Courant (CC)]
		The CC rule is the committee rule that selects the committee $W \in \mathcal{C}_k$ that maximizes the CC score:
		\[
		\mathit{score}_{CC}(W) = \sum_{i \in N} \min(1, \left| P_i \cap W\right|)
		\]
		Ties between committees with the same score are broken by a lexicographic order.
	\end{definition}

	Let us revisit Example \ref{ex:representation_failure} under the CC rule. The maximum possible CC score is $10$, which is achieved if every single voter is represented by at least one candidate. Any committee of size $5$ containing $\{c_3, c_4, c_5\}$ covers all voters and thus achieves this maximum score. Following the lexicographic tie-breaking rule, the winning committee is $W_{CC} = \{c_1, c_2, c_3, c_4, c_5\}$. At first glance, this outcome is excellent: we have achieved perfect representation, ensuring that every voter has an approved candidate on the committee. However, a deeper structural problem emerges. Once a voter is represented by a single candidate, their ballot loses all further influence. The cohesive majority group of six voters, constituting $60\%$ of the electorate, successfully elected $c_5$. Yet, despite their overwhelming numbers, they receive only this single seat, failing to secure a majority in the committee. This highlights a critical flaw: while basic representation is guaranteed, the size of a group is completely ignored once their first representative is secured. This leaves us needing a middle ground, a rule that protects minority voices, but also ensures that the number of seats a group receives appropriately reflects its actual size within the population.

	The natural middle ground between the extremes of AV and CC was proposed by Thiele himself as early as the late 19th century~\cite{thiele1895pav} and is known as Proportional Approval Voting (PAV). PAV recognizes that while securing a voter's first representative is crucial, subsequent representatives still hold value, but at a decreasing rate. In the Thiele framework, PAV is defined by the harmonic weight function, where the $j$-th approved candidate contributes exactly $1/j$ to a voter's satisfaction.

	\begin{definition}[Proportional Approval Voting (PAV)]
		The PAV rule is the committee rule that selects the committee $W \in \mathcal{C}_k$ that maximizes the PAV score:
		\[
		score_{PAV}(W) = \sum_{i \in N} H(\left| P_i \cap W\right| )
		\]
		where $H(x)$ is defined as the sum of the harmonic series up to $x$, i.e., $H(x) = \sum_{j=1}^{x} \frac{1}{j}$. 	
		Ties between committees with the same score are broken by a lexicographic order.
	\end{definition}

	By naturally encouraging the selection of diverse committees that represent distinct cohesive groups, PAV provides a compelling baseline for proportional representation. However, as we will show later in this paper, it achieves limited results in terms of strategyproofness. To reduce its strategy-vulnerability, we introduce a new parameterized Thiele rule: \emph{$d$-Regressive Proportional Approval Voting ($d$-RPAV)}. The parameter $d \ge 1$ allows us to control this spectrum, where at the extreme of $d=1$, the rule perfectly coincides with standard PAV.
	
	\begin{definition}[$d$-Regressive Proportional Approval Voting ($d$-RPAV)]
		The $d$-RPAV rule is the committee rule that selects the committee $W \in \mathcal{C}_k$ that maximizes the $d$-RPAV score:
		\[
		score_{d-RPAV}(W) = \sum_{i \in N} H_d(\left| P_i \cap W\right| )
		\]
		where $H_d(x)$ is defined as the sum of the following series up to $x$: $H_d(x) = \sum_{j=1}^{x} \frac{d}{j+d-1}$. 	
		Ties between committees with the same score are broken by a lexicographic order.
	\end{definition}

	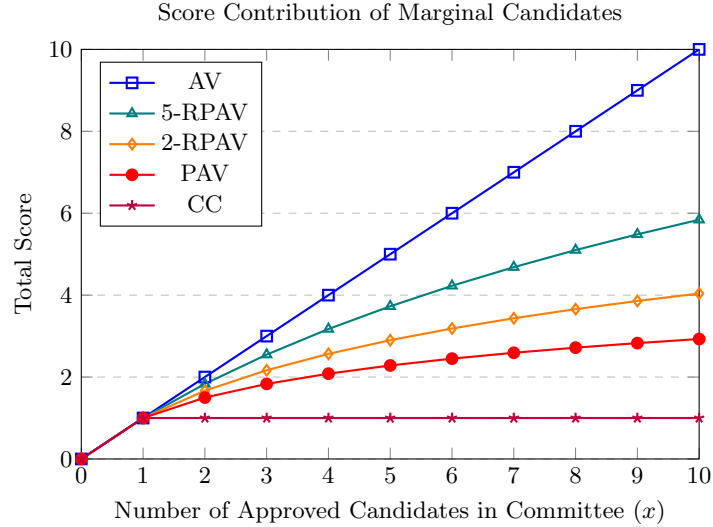
\begin{figure}[ht]
		\centering
		\begin{tikzpicture}
			\begin{axis}[
				title={Score Contribution of Marginal Candidates},
				xlabel={Number of Approved Candidates in Committee ($x$)},
				ylabel={Total Score},
				xmin=0, xmax=10,
				ymin=0, ymax=10,
				xtick={0,1,...,10},
				ytick={0,2,4,6,8,10},
				legend pos=north west,
				ymajorgrids=true,
				grid style=dashed,
				width=0.8\textwidth,
				height=7cm
				]
				
				\addplot[
				color=blue,
				mark=square,
				domain=0:10,
				samples=11,
				thick
				]
				{x};
				\addlegendentry{AV}
				
				\addplot[
				color=teal,
				mark=triangle,
				thick
				]
				coordinates {
					(0,0) (1,1) (2,1.833) (3,2.548) (4,3.173) (5,3.728)
					(6,4.228) (7,4.683) (8,5.100) (9,5.484) (10,5.841)
				};
				\addlegendentry{5-RPAV}
				
				\addplot[
				color=orange,
				mark=diamond,
				thick
				]
				coordinates {
					(0,0) (1,1) (2,1.667) (3,2.167) (4,2.567) (5,2.900)
					(6,3.186) (7,3.436) (8,3.658) (9,3.858) (10,4.040)
				};
				\addlegendentry{2-RPAV}
				
				\addplot[
				color=red,
				mark=*,
				thick
				]
				coordinates {
					(0,0) (1,1) (2,1.5) (3,1.833) (4,2.083) (5,2.283)
					(6,2.45) (7,2.593) (8,2.718) (9,2.829) (10,2.929)
				};
				\addlegendentry{PAV}
				
				\addplot[
				color=purple,
				mark=star,
				thick
				]
				coordinates {
					(0,0) (1,1) (2,1) (3,1) (4,1) (5,1)
					(6,1) (7,1) (8,1) (9,1) (10,1)
				};
				\addlegendentry{CC}
				
			\end{axis}
		\end{tikzpicture}
		\caption{The diminishing marginal returns of the $d$-RPAV rules compared to AV and CC. Higher values of $d$ steepen the curve towards standard Approval Voting.}
		\label{fig:drpav-interpolation}
	\end{figure}

	The term \emph{regressive} in our rule's name derives from the broader categorization of Thiele rules. In the literature \cite{lackner2023multi}\footnote{For specific details, see Section 4.5, pp.~65-71.} \cite{Jaworski2022PhragmnRF}, rules that fall between PAV and CC are often associated with degressive proportionality, as they give heavier weight to minority representation. Conversely, the spectrum between PAV and AV is characterized by regressive proportionality, tending more towards the majority-favoring nature of AV. Since standard PAV represents exact proportionality, our parameter $d > 1$ navigates the regressive space.
	
	To visually summarize this theoretical spectrum, Figure \ref{fig:drpav-interpolation} illustrates how the choice of the weight function affects the cumulative score as a voter's representation increases. It clearly demonstrates the interpolation from the extreme of CC, through the exact proportionality of PAV, and along the regressive path of $d$-RPAV towards AV.
	
	With the mechanisms established, in the following sections we will formally define the axiomatic properties regarding proportional representation and strategy-vulnerability, and analyze how these rules measure up against them.
	
	\subsection{Notations}
	
	Given a voter $i \in N$ and a subset of voters $F \subseteq N \setminus \{i\}$, we denote by $\overline{F}$ the set of remaining voters, defined as $\overline{F} \coloneqq N \setminus ( \{i\} \cup F$ ). When applying a committee rule to the complete profile, we write it as $\mathcal{R}(P_i, P_{-i})$.
	
	To model manipulation under partial information, we partition the rest of the electorate into two distinct groups: those whose ballots are known to the manipulator, and those whose ballots remain unknown. Given a voter $i \in N$ and a subset of voters $F \subseteq N \setminus \{i\}$ representing the \emph{familiar} (or foreknown) voters, we denote by $\overline{F}$ the set of remaining unknown voters, defined as $\overline{F} \coloneqq N \setminus ( \{i\} \cup F )$. This explicit partition allows us to analyze risk-avoiding strategies, where the manipulator $i$ relies on the known ballots in $P_F$ to ensure their strategy is strictly safe against any possible realization of the unknown ballots in $P_{\overline{F}}$. Similarly, we use the notation $\mathcal{R}(P_i, P_F, P_{\overline{F}})$ to represent the application of the rule to the preferences of all players under this partition.

	\subsection{Definitions of Strategy-Proofness and Vulnerability}
	\label{subsec:strategyproof}
	
	We begin by formalizing the standard notion of strategy-proofness in approval-based multi-winner elections. The most fundamental assumption is that a voter's utility is derived directly from the number of approved candidates that are elected \cite{aziz2015computational,peters2018}. Thus, a rule is considered vulnerable if a voter can alter their ballot to strictly increase this number.
	
	\begin{definition}[Cardinality Strategy-Vulnerability (C-SVul)]
		A committee rule $\mathcal{R}$ satisfies \emph{cardinality strategy-vulnerability} (C-SVul) if there exists a voter $i \in N$, an approval ballot $P_i \subset C$, and an alternative ballot  $A_i \subset C$ such that $A_i \neq P_i$ and  there exists $P_{-i}$ such that:
		\[
		|\mathcal{R}(A_i, P_{-i}) \cap P_i| > |\mathcal{R}(P_i, P_{-i}) \cap P_i|
		\]
	\end{definition}
	
	\begin{definition}[Cardinality Strategy-Proofness (C-SProof)]
		A committee rule satisfies \emph{cardinality strategy-proofness} (C-SProof)) if it does not satisfy C-SVul.
	\end{definition}
	
	A prevalent manipulation in approval-based rules is \emph{free-riding}, or \emph{dropping candidates} \cite{hylland1992proportional,schulze2004free}, where a voter omits safe candidates to concentrate voting weight on marginal ones. Restricting the manipulator exclusively to this behavior relaxes the strategy-proofness requirement, capturing a highly natural strategic threat:
	
	\begin{definition}[Cardinality Strategy-Vulnerability under Dropping Candidates (C-SVul-DC)]
		A committee rule $\mathcal{R}$ satisfies \emph{cardinality strategy-vulnerability under dropping candidates} (C-SVul-DC) if there exists a voter $i \in N$, an approval ballot $P_i \subset C$, and an alternative ballot $A_i \subset P_i$ such that there exists $P_{-i}$ such that:
		\[
		|\mathcal{R}(A_i, P_{-i}) \cap P_i| > |\mathcal{R}(P_i, P_{-i}) \cap P_i|
		\]
	\end{definition}
	\begin{definition}[Cardinality Strategy-Proofness under Dropping Candidates (C-SProof-DC)]
		A committee rule satisfies \emph{cardinality strategy-proofness under dropping candidates} (C-SProof-DC) if it does not satisfy C-SVul-DC.
	\end{definition}
	
	Cardinality-based definitions demand robust immunity, as they consider any increase in the size of the approved intersection a successful manipulation. However, as recognized in the literature \cite{aziz2019fair,peters2018}, voters' preferences over entire committees can be complex. We relax the strategy-proofness requirement by adopting the standard \emph{superset} framework: a manipulation is beneficial only if the new outcome contains every previously secured candidate, plus at least one more.

	\begin{definition}[Superset Strategy-Vulnerability (S-SVul)]
		A committee rule $\mathcal{R}$ satisfies \emph{superset strategy-vulnerability} (S-SVul) if there exists a voter $i \in N$, an approval ballot $P_i \subset C$, and an alternative ballot $A_i \subset C$ such that $A_i \neq P_i$ and there exists $P_{-i}$ such that:
		\[
		\mathcal{R}(A_i, P_{-i}) \cap P_i \supsetneq \mathcal{R}(P_i, P_{-i}) \cap P_i
		\]
	\end{definition}
	\begin{definition}[Superset Strategy-Proofness (S-SProof)]
		A committee rule satisfies \emph{superset strategy-proofness} (S-SProof) if it does not satisfy S-SVul.
	\end{definition}
	
	Combining the superset framework with the practical strategy of dropping candidates yields a targeted, relaxed requirement. This specific combination forms the structural baseline that we will now extend into the risk-avoiding domain.
	
	\begin{definition}[Superset Strategy-Vulnerability under Dropping Candidates (S-SVul-DC)]
		\sloppy
		A committee rule $\mathcal{R}$ satisfies \emph{superset strategy-vulnerability under dropping candidates} (S-SVul-DC) if there exists a voter $i \in N$, an approval ballot $P_i \subset C$, and an alternative ballot $A_i \subset P_i$ such that there exists $P_{-i}$ such that:
		\[
		\mathcal{R}(A_i, P_{-i}) \cap P_i \supsetneq \mathcal{R}(P_i, P_{-i}) \cap P_i
		\]
	\end{definition}
	\begin{definition}[Superset Strategy-Proofness under Dropping Candidates (S-SProof-DC)]
		A committee rule satisfies \emph{superset strategy-proofness under dropping candidates} (S-SProof-DC) if it does not satisfy S-SVul-DC.
	\end{definition}
	
	The definitions above assume the manipulator has complete knowledge of all other votes, or is willing to take the risk and act as if they do. To capture real-world uncertainty, we adapt the concept of risk-avoiding truthfulness (RAT) \cite{bu2023existence} to ABC rules. A risk-avoiding voter will only deviate if the strategy is strictly safe, never yielding a strictly worse outcome, while offering a potential improvement.
	
	\begin{definition}[Superset Risk-Avoiding Strategy-Vulnerability under Dropping Candidates (S-RASVul-DC)]
		A committee rule $\mathcal{R}$ satisfies \emph{superset risk-avoiding strategy-vulnerability under dropping candidates} (S-RASVul-DC) if there exists a voter $i \in N$, an approval ballot $P_i \subset C$, and an alternative ballot $A_i \subset P_i$ such that:
		\begin{enumerate}
			\item For every $P_{-i}$, it holds that:
			\[
			\mathcal{R}(A_i, P_{-i}) \cap P_i \supseteq
			\mathcal{R}(P_i, P_{-i}) \cap P_i 
			\]
			\item There exists $P_{-i}$ such that:
			\[
			\mathcal{R}(A_i, P_{-i}) \cap P_i \supsetneq \mathcal{R}(P_i, P_{-i}) \cap P_i
			\]
		\end{enumerate}
	\end{definition}
	\begin{definition}[Superset Risk-Avoiding Strategy-Proofness under Dropping Candidates (S-RASProof-DC)]
		A committee rule satisfies \emph{superset risk-avoiding strategy-proofness under dropping candidates} (S-RASProof-DC) if it does not satisfy S-RASVul-DC.
	\end{definition}

	To bridge the gap between zero knowledge and complete information, we parameterize this framework based on $f$, the exact number of foreknown ballots. Adapting the RAT-degree \cite{hartman2025} to this setting, we can precisely quantify the information required to safely manipulate.
	
	The following definitions are for $0 \le f \le n-1$:
	
	\begin{definition}[$f$-Superset Risk-Avoiding Strategy-Vulnerability under Dropping Candidates ($f$-S-RASVul-DC)]
		\label{def:f-s-rasvul-dc}
		A committee rule $\mathcal{R}$ satisfies \emph{$f$-superset risk-avoiding strategy-vulnerability under dropping candidates} ($f$-S-RASVul-DC) if there exists a voter $i \in N$ and a subset of voters $F \subseteq N \setminus \{i\}$ with $|F| = f$, such that for this voter there exists an approval ballot $P_i \subset C$ and an alternative ballot $A_i \subset P_i$ satisfying:
		\begin{enumerate}
			\item For every $P_{\overline{F}}$, it holds that:
			\[
			\mathcal{R}(A_i, P_F, P_{\overline{F}}) \cap P_i \supseteq \mathcal{R}(P_i, P_F, P_{\overline{F}}) \cap P_i
			\]
			\item There exists $P_{\overline{F}}$ such that:
			\[
			\mathcal{R}(A_i, P_F, P_{\overline{F}}) \cap P_i \supsetneq \mathcal{R}(P_i, P_F, P_{\overline{F}}) \cap P_i
			\]
		\end{enumerate}
	\end{definition}
	\begin{definition}[$f$-Superset Risk-Avoiding Strategy-Proofness under Dropping Candidates ($f$-S-RASProof-DC)]
		\sloppy
		A committee rule satisfies \emph{$f$-superset risk-avoiding strategy-proofness under dropping candidates} ($f$-S-RASProof-DC) if it does not satisfy $f$-S-RASVul-DC.
	\end{definition}
	
	In Definition \ref{def:f-s-rasvul-dc}, the known voters in $F$ serve as a fixed baseline for the manipulator's calculation. Condition 1 acts as a \emph{safety condition}: the manipulation must never harm the voter under \emph{any} possible realization of the unknown ballots in $\overline{F}$. Condition 2 requires the manipulation to be strictly \emph{profitable} in at least one specific scenario for $\overline{F}$. A natural consequence of this formulation is that gaining more information can never reduce a manipulator's ability to act safely.
	
	\begin{observation}
		\label{obs:f_increment_f-s-rasvul-dc}
		If a committee rule satisfies $f$-S-RASVul-DC for some $0 \le f \le n-2$, then it also satisfies $(f+1)$-S-RASVul-DC.
	\end{observation}
	
	This property follows directly from Definition \ref{def:f-s-rasvul-dc}. If a manipulation is safe and strictly profitable across all possible combinations of unknown ballots in $\overline{F}$, moving one voter from $\overline{F}$ to the known set $F$ (fixing their ballot to the specific scenario that yields the profit) merely restricts this set of unknown profiles. Thus, the manipulation trivially remains safe, and the strictly profitable outcome remains attainable.
	
	Observation \ref{obs:f_increment_f-s-rasvul-dc} naturally leads to the following definition.
	
	\begin{definition}[Degree of Superset Risk-Avoiding Strategy-Proofness under Dropping Candidates (S-RASProof-DC Degree)]
		The \emph{degree of superset risk-avoiding strategy-proofness under dropping candidates} (\emph{S-RASProof-DC Degree}) is defined as the minimum integer $f$ for which the rule satisfies $f$-S-RASVul-DC. Furthermore, if the rule does not satisfy $f$-S-RASVul-DC for any such $f$ (i.e., it satisfies S-SProof-DC), its degree is defined to be $n$.
		
	\end{definition}
	
	\subsection{Definitions of Proportionality}

	\begin{definition}[Proportionality]
		\label{def:proportionality}
		A committee rule $\mathcal{R}$ is defined as \emph{proportional} if, for any profile $P$ in which there are $\left\lceil \frac{n}{k} \right\rceil$ voters whose approval ballot is exactly $\{c\}$, it holds that $c \in \mathcal{R}(P)$.
	\end{definition}
	
	While we utilize this minimal notion of proportionality for our analysis, the literature frequently considers slightly stronger definitions. A prominent example is \emph{Justified Representation (JR)} \cite{Aziz2017}.
	
	\begin{definition}[Justified Representation (JR)]
		A committee rule $\mathcal{R}$ satisfies \emph{Justified Representation (JR)} if, for any profile $P$ and any group of voters $N^* \subseteq N$ of size $|N^*| \ge \left\lceil \frac{n}{k} \right\rceil$ who mutually approve at least one common candidate (i.e., $\bigcap_{i \in N^*} P_i \neq \emptyset$), there exists at least one voter $i \in N^*$ such that $P_i \cap \mathcal{R}(P) \neq \emptyset$.
	\end{definition}
	
	Definition \ref{def:proportionality} is a weaker variant of JR, as it strictly guarantees the election of the common candidate only when the cohesive group's ballots are exclusively limited to that single candidate. It is straightforward to observe that JR strictly implies Definition \ref{def:proportionality}'s proportionality requirement. If a cohesive group of size $\lceil \frac{n}{k} \rceil$ approves exactly $\{c\}$, JR requires that at least one voter in this group has an approved candidate in the winning committee. Since the only candidate approved by these voters is $c$, it follows directly that $c$ must be elected.
	
	According to previous literature \cite{Aziz2017}, the PAV rule satisfies much stronger proportionality properties than the one defined here, such as Extended Justified Representation (EJR). In particular, it satisfies JR, meaning it certainly satisfies our Definition \ref{def:proportionality}'s proportionality requirement.
	
	Since our generalized $d$-RPAV rule intentionally relaxes exact proportionality to achieve stronger strategy-proofness, we require a correspondingly relaxed, parameterized definition to evaluate its representation guarantees.
	
	\begin{definition}[$\frac{1}{d}$-Proportionality]
		\label{def:d_proportionality}
		A committee rule $\mathcal{R}$ is defined as \emph{$\frac{1}{d}$-proportional} if, for any profile $P$ in which there are $\left\lceil \frac{dn}{k} \right\rceil$ voters whose approval ballot is exactly $\{c\}$, it holds that $c \in \mathcal{R}(P)$.
	\end{definition}

	The stronger JR definition has also been generalized to a parameterized version. Recently, Becker and Frank introduced \emph{$\alpha$-Justified Representation ($\alpha$-JR)} to study instance-dependent quotas \cite{becker2026}. While their work focuses on values of $\alpha < 1$ to strengthen the JR requirement by demanding representation for smaller cohesive groups, their generalized definition seamlessly accommodates our goal of relaxing proportionality.
	
	\begin{definition}[$\alpha$-Justified Representation ($\alpha$-JR) \cite{becker2026}]
		A committee rule $\mathcal{R}$ satisfies \emph{$\alpha$-Justified Representation ($\alpha$-JR)} if, for any profile $P$ and any group of voters $N^* \subseteq N$ of size $|N^*| \ge \left\lceil \frac{\alpha n}{k} \right\rceil$ who mutually approve at least one common candidate (i.e., $\bigcap_{i \in N^*} P_i \neq \emptyset$), there exists at least one voter $i \in N^*$ such that $P_i \cap \mathcal{R}(P) \neq \emptyset$.
	\end{definition}
	
	In our context, by setting the parameter $\alpha = d$ (for $d \ge 1$), it is straightforward to observe that $\alpha$-JR strictly implies our minimal notion of $\frac{1}{d}$-proportionality (Definition \ref{def:d_proportionality}). If a cohesive group of size $\lceil \frac{dn}{k} \rceil$ approves exactly $\{c\}$, $\alpha$-JR requires that at least one voter in this group has an approved candidate in the winning committee. Since the only candidate approved by these voters is $c$, it follows directly that $c$ must be elected.

	\section{The S-RASProof-DC Degree of the Standard PAV Rule}
	\label{sec:pav-rule}
	
	In this section, we present our main results regarding the vulnerability and strategy-proofness of approval-based committee rules under candidate-dropping manipulation. We begin by evaluating the standard PAV rule, establishing exact thresholds for both its vulnerability ($f$-S-RASVul-DC) and its corresponding proofness ($f$-S-RASProof-DC). 
	
	We first establish a bound that guarantees the vulnerability of the standard PAV rule to candidate-dropping manipulation.
	
	\newcommand{\textThmPavVul}
	{
		The PAV rule satisfies $f$-S-RASVul-DC for $f = \left\lceil \frac{n}{k} \right\rceil$.
	}
	\begin{theorem}
		\label{thm:pav-vul}
		\textThmPavVul
	\end{theorem}
	
	\begin{proofsketch}
		The intuition for this vulnerability relies on the proportional nature of the PAV rule, combined with its diminishing marginal contributions. 
		
		Suppose the manipulator's true preference is $\{c_1, c_2\}$, and they know that exactly $f = \lceil \frac{n}{k} \rceil$ voters exclusively support $c_1$. Because a quota of $\lceil \frac{n}{k} \rceil$ voters mathematically guarantees a candidate's election under PAV, the known voters perfectly secure $c_1$'s position. The manipulator can therefore safely drop $c_1$ from their ballot without risking $c_1$'s exclusion from the winning committee, satisfying the safety condition (Condition~1 of Definition~\ref{def:f-s-rasvul-dc}).
		
		To see why this manipulation is actively beneficial, we compare the marginal score contributed to $c_2$ under both strategies. Under truthful voting, the guaranteed election of $c_1$ subjects the manipulator to a diminishing return, reducing their marginal contribution for $c_2$ to just a half. Conversely, by strategically dropping $c_1$, the manipulator avoids this penalty. They free-ride on the known voters' support for $c_1$ while contributing a full marginal score of 1 to $c_2$.
		
		To demonstrate this potential gain, we construct a specific threshold profile for the remaining unknown voters where $c_2$ is locked in a tight race for the final committee seat. In this delicate scenario, the reduced marginal score of $1/2$ from the truthful ballot is insufficient to elect $c_2$. However, the full marginal score of $1$ from the manipulated ballot successfully pushes $c_2$ into the elected committee alongside $c_1$. As a result, the manipulated outcome includes strictly more of the voter's preferred candidates, effectively satisfying Condition~2 of Definition~\ref{def:f-s-rasvul-dc}. The explicit algebraic proof of this scenario is detailed in Appendix~\ref{proof:pav-vul}.
	\end{proofsketch}

	While the previous theorem establishes the threshold for vulnerability ($f$-S-RASVul-DC), our next result proves the corresponding bound for the rule's proofness ($f$-S-RASProof-DC).
	
	\newcommand{\textThmPavProofness}
	{
		The PAV rule satisfies $f$-S-RASProof-DC for $f = \lfloor \frac{n}{k+1} \rfloor - 1$.
	}
	\begin{theorem}
		\label{thm:pav-proofness}
		\textThmPavProofness
	\end{theorem}
	
	\begin{proofsketch}
		
		The formal proof of this theorem follows directly as a special case of our generalized result for the $d$-RPAV rule, which we establish later in Theorem \ref{thm:drpav-proofness}. By observing that standard PAV is exactly equivalent to $d$-RPAV for the parameter $d=1$, we can substitute $d=1$ into the generalized size of the known voters group, $f = \lfloor \frac{dn}{k+2d-1} \rfloor - 1$, to directly derive the required size of $F$ for standard PAV: $f = \lfloor \frac{n}{k+1} \rfloor - 1$.
		
		Nevertheless, while the theoretical guarantee relies on the general proof, that proof is highly intricate. To create a gradual buildup and provide better intuition, we sketch here the manipulator's worst-case scenario for standard PAV. We assume the manipulator attempts to safely drop a candidate $c'$ while knowing the ballots of exactly $f$ other voters. The optimal state for this manipulation occurs when all $f$ known voters exclusively approve $c'$, providing maximal guaranteed support.
		
		To prove the manipulation is unsafe, we construct a specific profile for the remaining unknown voters that exploits the diminishing marginal contributions inherent to the PAV rule. In this profile, all unknown voters approve a base set of $k-1$ candidates, while exactly $f \cdot k + 1$ of them additionally approve an outsider candidate $c''$ (who is unapproved by the manipulator). Because the marginal PAV contribution for a $k$-th approved candidate is $1/k$, this subset provides a combined marginal score of $f + 1/k$ to $c''$. This generates a razor-thin margin over $c'$, who receives a marginal score of exactly $f$ from the known voters. Under this delicate balance, the manipulator's truthful approval of $c'$ serves as the decisive weight. By dropping $c'$ to free-ride on the known supporters, the manipulator inadvertently removes this critical weight, causing the outsider $c''$ to win instead. Because a preferred candidate is lost, the manipulated outcome's intersection with the true preferences fails to contain the truthful outcome's intersection. This violates the safety condition (Condition~1 of Definition~\ref{def:f-s-rasvul-dc}). A comprehensive algebraic analysis of this worst-case profile is provided in Appendix~\ref{proof:pav-proofness}.
	\end{proofsketch}
	
	While Theorem~\ref{thm:pav-vul} establishes vulnerability ($f$-S-RASVul-DC) for $f = \left\lceil \frac{n}{k} \right\rceil$ and Theorem~\ref{thm:pav-proofness} guarantees proofness ($f$-S-RASProof-DC) up to $f = \left\lfloor \frac{n}{k+1} \right\rfloor - 1$, a gap naturally remains between these two bounds. A fundamental question is whether the proofness threshold can be further improved. The following example demonstrates that the bound established in Theorem~\ref{thm:pav-proofness} is, in fact, tight. Specifically, we present a concrete instance where exactly $f = \left\lfloor \frac{n}{k+1} \right\rfloor$ known voters are sufficient to enable a manipulation that is both safe and strictly profitable (satisfying both Condition~1 and Condition~2 of Definition~\ref{def:f-s-rasvul-dc}). Consequently, achieving a general guaranteed S-RASProof-DC degree of this value is impossible without imposing further restrictive assumptions.
	
	\begin{example}
		\label{ex:tightness}
		Consider an election with $n=6$ voters, a committee size of $k=2$, and a candidate set $C = \{c_1, c_2, c_3\}$. We assume a lexicographic tie-breaking rule $c_1 \succ c_2 \succ c_3$, which induces the following priority over committees of size two: $\{c_1, c_2\} \succ \{c_1, c_3\} \succ \{c_2, c_3\}$.
		
		Here, the threshold in question is $f = \left\lfloor \frac{6}{3} \right\rfloor = 2$. Let the set of known voters $F$ consist of exactly $f=2$ voters, both reporting the ballot $\{c_1\}$. Consider a manipulating voter $i$ with the truthful preference $P_i = \{c_1, c_2\}$, who contemplates the manipulated ballot $A_i = \{c_2\}$. The remaining unknown voters, denoted by $\overline{F}$, consist of $f' = 3$ voters.
		
		Through a careful verification of the possible voting profiles for $\overline{F}$, it can be confirmed that the committee $\{c_2, c_3\}$ never wins under the manipulated ballot. Because of this, the manipulator will never cause $c_1$ to be dropped from the winning committee by reporting $A_i$. Furthermore, since the manipulated ballot $A_i = \{c_2\}$ concentrates the voter's approval solely on $c_2$, reporting $A_i$ can only strengthen the relative standing of committees containing $c_2$. Therefore, the manipulator risks losing neither $c_1$ nor $c_2$, meaning this manipulation is \textbf{safe} (satisfying Condition~1 of Definition~\ref{def:f-s-rasvul-dc}).
		
		To demonstrate that this manipulation is also \textbf{profitable} (satisfying Condition~2 of Definition~\ref{def:f-s-rasvul-dc}), consider a specific voting profile for the unknown voters where two voters in $\overline{F}$ vote for $\{c_1\}$ and the third votes for $\{c_3\}$. Under truthful reporting ($P_i = \{c_1, c_2\}$), the committee $\{c_1, c_3\}$ achieves the strictly highest score and wins the election. However, under the manipulated ballot ($A_i = \{c_2\}$), the score of $\{c_1, c_3\}$ drops, creating a tie between $\{c_1, c_2\}$ and $\{c_1, c_3\}$. Following the lexicographic tie-breaking rule ($\{c_1, c_2\} \succ \{c_1, c_3\}$), the committee $\{c_1, c_2\}$ becomes the unique winner. Consequently, the set of the manipulator's truthfully approved candidates in the winning committee strictly expands from $\{c_1\}$ to $\{c_1, c_2\}$, as required.
		
	\end{example}
	
	\section{Relaxing Proportionality for a Higher S-RASProof-DC Degree: The $d$-RPAV Rule}
	\label{sec:drpav-rule}
	
	As established in Theorem~\ref{thm:pav-vul}, the standard PAV rule achieves an S-RASProof-DC degree of at most $\left\lceil \frac{n}{k} \right\rceil$. To increase this degree, we explore the following trade-off: relaxing the strict proportionality requirement to a weaker version in order to achieve a higher S-RASProof-DC degree. This approach leads us to the parameterized $d$-RPAV family, which allows us to tune this trade-off at various levels by choosing the parameter $d$.

	Before determining the enhanced S-RASProof-DC degree of this generalized rule, we must first establish its exact proportionality guarantees. Since we established that setting $\alpha = d$ yields a definition ($\alpha$-JR) that strictly implies our definition of $\frac{1}{d}$-proportionality, we will proceed to prove that the $d$-RPAV rule satisfies this stronger definition directly, thereby proving that it satisfies $\frac{1}{d}$-proportionality as well.
	
	\newcommand{\textThmDrpavAlphaJR}
	{
		The $d$-RPAV rule satisfies $\alpha$-JR where $\alpha = d$ and $d \ge 1$.
	}
	\begin{theorem}
		\label{thm:drpav-alpha-jr}
		\textThmDrpavAlphaJR
	\end{theorem}

	\begin{proofsketch}
		The proof relies on a swap argument, showing that if a cohesive group of voters is left entirely unrepresented, the winning committee cannot be optimal.
		
		Assume towards a contradiction that a sufficiently large group of voters $N^*$, who mutually approve at least one common candidate, receives no representation in the winning committee. We analyze the \emph{marginal contribution} of each committee member, defined as the score the committee would lose if that member were removed. Since $N^*$ has no representatives, the committee's score relies exclusively on the remaining voters. We establish a mathematical upper bound on how much these remaining voters can possibly contribute to the sum of marginal contributions.
		
		By the pigeonhole principle, this bound forces the existence of at least one "weak" candidate in the committee, whose marginal contribution is relatively small. If we swap this candidate with a candidate mutually approved by $N^*$, the score lost is minimal. Conversely, because the voters in $N^*$ previously had zero representatives, adding their mutually approved candidate guarantees a large and immediate increase in the score. We show that this gain strictly outweighs the loss, resulting in a new committee with a higher overall score. This contradicts the fundamental definition of the $d$-RPAV rule, which must select a committee that maximizes the total score. The complete formal proof is detailed in Appendix~\ref{proof:drpav-alpha-jr}.
	\end{proofsketch}

	Finally, having established that the $d$-RPAV rule satisfies $\alpha$-JR (and thus $\frac{1}{d}$-proportionality), we turn to its strategic immunity. Since standard PAV is a special case of $d$-RPAV (where $d=1$) and Theorem~\ref{thm:pav-vul} already demonstrates its vulnerability given a sufficiently large subset of known voters, our analysis of the generalized rule focuses directly on quantifying its enhanced proofness. To this end, our next theorem establishes the generalized bound on the S-RASProof-DC degree for the $d$-RPAV rule, demonstrating how tuning the parameter $d$ explicitly yields a higher S-RASProof-DC degree.

	\newcommand{\textThmDrpavProofness}
	{
		The $d$-RPAV rule satisfies $f$-S-RASProof-DC for $f = \left\lfloor \frac{dn}{k+2d-1}\right\rfloor - 1 $.
	}
	\begin{theorem}
		\label{thm:drpav-proofness}
		\textThmDrpavProofness
	\end{theorem}
	
	\begin{proofsketch}
		Assume toward a contradiction that the mechanism is vulnerable: some voter $i$, having observed a set $F$ of $f = \left\lfloor \frac{dn}{k+2d-1}\right\rfloor - 1$ other ballots, can safely manipulate by reporting $A_i \subsetneq P_i$ in place of their true ballot $P_i$, regardless of how the remaining voters $\overline{F}$ report. At the heart of deriving a contradiction lies a single candidate, $c'$: one that $i$ approves under $P_i$ but drops under $A_i$. This is precisely the candidate Condition~1 of Definition~\ref{def:f-s-rasvul-dc} is concerned with, the manipulation is unsafe exactly when it causes $i$ to lose, under $A_i$, some candidate from $P_i$ that they would have won by reporting truthfully. The entire proof therefore reduces to a single task: exhibit a report for $\overline{F}$ under which $c'$ is elected when $i$ reports truthfully, but not when $i$ manipulates.
		
		To this end, we pit two committees against each other. The first, $W'$, contains $c'$; the second, $W''$, agrees with $W'$ on every other seat but replaces $c'$ with some other candidate $c''$. The identity of the $k-1$ shared candidates, and of $c''$ itself, plays no role in the argument beyond this construction, what matters is only that $c' \in W'$ and $c' \notin W''$. We arrange, through a precise choice of how $\overline{F}$ votes, for a committee containing $c'$ (such as $W'$) to win under the truthful profile and a committee without $c'$ (such as $W''$) to win under the manipulated one, so that $i$ loses $c'$ precisely because of the manipulation.
		
		This plan is carried out in five steps. \textbf{Step 1} performs the construction concretely: given the specific $F$ and voter $i$'s two ballots, it identifies the actual candidates filling $W'$ and $W''$, and lays out the general shape of a vote for $\overline{F}$ that will expose this risk to voter $i$. \textbf{Step 2} pins down the exact numbers within that shape so that $W'$ defeats $W''$ under the truthful profile, and $W''$ defeats $W'$ under the manipulated one. \textbf{Step 3} departs from this comparison: rather than pitting committees against each other, it confirms that the vote counts Step 2 calls for are actually feasible, that $\overline{F}$, of size $f' = n-f-1$, is large enough to place $c'$ precisely on the seam between winning and losing that the rest of the proof requires. Unlike the other steps, which compare scores between committees, Step 3 is purely arithmetic: it confirms that the required $x'$ and $x''$ exist at all, given the bound on $f$. \textbf{Step 4} then rules out the remaining ``hybrid'' committee $W^*$, which combines pieces of $W'$ and $W''$ by containing both $c'$ and $c''$. \textbf{Step 5} rules out every other conceivable committee via a general candidate-swapping argument, so that the victories established in Step 2 are not artifacts of an incomplete comparison.
		
		Taken together, these steps exhibit a profile $P_{\overline{F}}$ for which $c' \notin \mathcal{R}(A_i, P_F, P_{\overline{F}})$ while $c' \in \mathcal{R}(P_i, P_F, P_{\overline{F}})$. Since $c' \in P_i$, this means $\mathcal{R}(A_i, P_F, P_{\overline{F}}) \cap P_i \supseteq \mathcal{R}(P_i, P_F, P_{\overline{F}}) \cap P_i$ fails: by manipulating, voter $i$ loses a candidate they would have won by reporting truthfully. This violates Condition~1 of Definition~\ref{def:f-s-rasvul-dc}, contradicting our counter assumption that the manipulation was safe. The complete formal proof is detailed in Appendix~\ref{proof:drpav-proofness}.
	\end{proofsketch}

	\section{Conclusions and Future Work}
	\label{sec:conclusions}
	
	Peters' impossibility says a proportional ABC rule must be
	manipulable~\cite{peters2018}; it does not say how much a voter must know
	before manipulation is worth attempting. We answered that question for PAV:
	safe free-riding becomes possible at $\lceil n/k \rceil$ known ballots and is never safe
	at or below $\lfloor n/(k+1) \rfloor - 1$, a bound that cannot be raised in general.
	Both thresholds scale as a $1/k$ fraction of the electorate, so PAV's practical
	robustness erodes exactly as committees grow. We then showed that a controlled
	concession on proportionality buys robustness at a known rate: $d$-RPAV
	guarantees $\alpha$-JR for $\alpha = d$ while resisting safe free-riding up to
	$\lfloor dn/(k+2d-1) \rfloor - 1$ known ballots. The conflict usually stated as
	a dichotomy thus admits a quantitative frontier, traversed in closed form by a
	single parameter.
	
	Three questions stand out. 
	First, our immunity guarantee tends to $n/2$ as $d \to \infty$. We conjecture
	that this represents a fundamental ceiling for any \emph{reasonable}
	multi-winner rule other than AV. If a voter observes a strict majority
	supporting a single candidate, dropping that candidate is always a safe
	manipulation, even under AV. Under AV, however, this safe drop yields no
	strict profit due to the lack of mutual interference between co-approved
		candidates. Yet under
	almost any other \emph{reasonable} rule, the recovered voting weight will
	likely make this drop strictly profitable under some profile completion.
	Formalizing this intuition into a general impossibility theorem, which would
	require precisely defining what constitutes a \emph{reasonable} rule through
	properties such as anonymity and non-separability, remains an intriguing
	direction for future work.
	Second, $d$-RPAV occupies a point on the trade-off,
	but we do not know whether it is an efficient one: the natural companion to
	Peters' qualitative impossibility is an upper bound on the S-RASProof-DC degree
	achievable by \emph{any} $\tfrac{1}{d}$-proportional rule, which would settle
	whether $d$-RPAV is optimal or merely convenient. 
	Third, the same analysis is open for Chamberlin-Courant (CC), as well as for polynomial-time proportional rules like sequential CC, sequential Phragm\'{e}n, and the Method of Equal Shares~\cite{lackner2023multi}. The local, budget-based, or sequential logic of the latter rules may yield qualitatively different degrees than PAV and $d$-RPAV, and they matter in practice given that computing optimal Thiele rules like PAV and CC is NP-hard~\cite{aziz2015computational}.
	Concurrent work~\cite{ai2026} has recently demonstrated vulnerabilities in sequential CC (referred to therein as Greedy AV) for $\lfloor n/2 \rfloor - 1$ known ballots, and in the Method of Equal Shares even for a single known ballot. However, their results rely on specific model assumptions (which differ slightly from ours) and rule variants. Consequently, establishing proofness guarantees and exact bounds for these rules under different assumptions remains an open question. Beyond these, the exact degree of PAV within our
	remaining gap, stronger manipulations than dropping under the cardinality
	criterion, the complexity of \emph{finding} a safe manipulation, and
	aggregate rather than ballot-level knowledge~\cite{hartman2025} all remain open.
	
	More broadly, impossibility theorems in social choice are proved against an
	omniscient adversary. Replacing that adversary with an informed but uncertain
	one turns such theorems from verdicts into measurements, which is what a
	designer choosing among rules actually needs.

	\bibliographystyle{splncs04}
	\bibliography{references}

	\newpage
	\appendix
	
	\section{Proofs for Section~\ref{sec:pav-rule}(\nameref{sec:pav-rule})}
	\label{app:proofs-pav}
	
	\phantomsection
	\label{proof:pav-vul}
	\noindent\textbf{Theorem~\ref{thm:pav-vul}.} \textit{\textThmPavVul}
	
	\begin{proof}
		For a voter $i \in N$, let the approval ballot be $P_i = \{c_1, c_2\}$ and the alternative ballot be $A_i = \{c_2\}$. Clearly, $A_i \subset P_i$. Let $F$ be a set of voters of size $\left\lceil \frac{n}{k} \right\rceil$ such that the approval ballot of every voter in $F$ is exactly $\{c_1\}$.
		
		By the definition of proportionality, for every $P_{\overline{F}}$, it holds that both $c_1 \in \mathcal{R}(P_i, P_F, P_{\overline{F}})$ and $c_1 \in \mathcal{R}(A_i, P_F, P_{\overline{F}})$.
		
		Furthermore, if $c_2 \in \mathcal{R}(P_i, P_F, P_{\overline{F}})$, then it can be said that the $score_{PAV}$ obtained for the winning committee containing $c_2$ under the profile $(P_i, P_F, P_{\overline{F}})$ prevailed over the $score_{PAV}$ assigned to any committee that does not include $c_2$. In this voting profile, the absence of $c_2$ from the winning committee would reduce the $score_{PAV}$ of the committee by at least $\frac{1}{2}$ (due to the contribution of voter $i$), and no unelected candidate could have contributed more to the $score_{PAV}$ than $c_2$ did.
		
		Now, under the voting profile $(A_i, P_F, P_{\overline{F}})$, everything remains the same except for $c_2$. However, the absence of $c_2$ would now reduce the $score_{PAV}$ by at least $1$ (due to the contribution of voter $i$, while the contributions of other voters remain the same). Therefore, no candidate who was not elected in $\mathcal{R}(P_i, P_F, P_{\overline{F}})$ can prevail over $c_2$, just as they could not prevail over it in the previous state. Thus, it must be the case that $c_2 \in \mathcal{R}(A_i, P_F, P_{\overline{F}})$. That is, we have shown that if $c_2 \in \mathcal{R}(P_i, P_F, P_{\overline{F}})$, then $c_2 \in \mathcal{R}(A_i, P_F, P_{\overline{F}})$.
		
		Recall that $P_i = \{c_1, c_2\}$. Therefore, if $c_2 \notin \mathcal{R}(P_i, P_F, P_{\overline{F}})$, then we get $\mathcal{R}(P_i, P_F, P_{\overline{F}}) \cap P_i = \{c_1\}$. As we have shown that $c_1 \in \mathcal{R}(A_i, P_F, P_{\overline{F}})$, it holds that $\mathcal{R}(A_i, P_F, P_{\overline{F}}) \cap P_i \supseteq \mathcal{R}(P_i, P_F, P_{\overline{F}}) \cap P_i$.
		
		On the other hand, if $c_2 \in \mathcal{R}(P_i, P_F, P_{\overline{F}})$, then $\mathcal{R}(P_i, P_F, P_{\overline{F}}) \cap P_i = \{c_1, c_2\}$. As we have shown in this case, $c_2 \in \mathcal{R}(A_i, P_F, P_{\overline{F}})$, and in any case $c_1 \in \mathcal{R}(A_i, P_F, P_{\overline{F}})$. Hence, in this case as well, it holds that $\mathcal{R}(A_i, P_F, P_{\overline{F}}) \cap P_i \supseteq \mathcal{R}(P_i, P_F, P_{\overline{F}}) \cap P_i$.
		
		This concludes the proof of Condition 1 in Definition \ref{def:f-s-rasvul-dc}.
		
		\bigskip
		
		For Condition 2 of Definition \ref{def:f-s-rasvul-dc}, it is easy to find a sub-profile $P_{\overline{F}}$ that places $c_2$ exactly on the winning threshold, such that if voter $i$ submits an approval ballot of exactly $\{c_2\}$, $c_2$ is pushed into the winning committee, whereas if he submits a ballot of $\{c_1, c_2\}$ leaves $c_2$ below the threshold to lose.
		
		Nevertheless, we will provide an explicit example of such a sub-profile. To ensure the manipulation remains safe (satisfying Condition~1 of Definition~\ref{def:f-s-rasvul-dc}), we maintain the exact same assumptions regarding the ballots $P_i$, $A_i$, and the set of known voters $F$ established at the beginning of the proof.
		
		Building on these assumptions, let $k=3$, and let us assume $n \ge 6$ to ensure non-trivial sizes for both $F$ and $\overline{F}$. Note that the number of the remaining voters is $|\overline{F}| = n - 1 - \left\lceil \frac{n}{k} \right\rceil$. Let $f' = |\overline{F}|$ denote this size. We assume, without loss of generality, that $f'$ is an even number (if it is not, we can assign one voter in $\overline{F}$ an approval ballot of exactly $\{c_1\}$ and distribute the remaining voters accordingly).
		
		We construct $P_{\overline{F}}$ such that there are $\frac{f'}{2} - 1$ voters whose approval ballot is exactly $\{c_1, c_2, c_4\}$, and $\frac{f'}{2} + 1$ voters whose approval ballot is exactly $\{c_1, c_3, c_4\}$.
		
		We can identify four possible winning committees. Let us denote them as follows: 
		$W_{-1} = \{c_2, c_3, c_4\}$, 
		$W_{-2} = \{c_1, c_3, c_4\}$, 
		$W_{-3} = \{c_1, c_2, c_4\}$, 
		and $W_{-4} = \{c_1, c_2, c_3\}$. 
		
		We know that $W_{-1}$ will not be the winning committee due to the proportionality property of the PAV rule. It is also highly unlikely that $W_{-4}$ will be the winner, but we will demonstrate this explicitly.
		
		First, let us examine the $score_{PAV}$ for the various possible committees under the profile $(P_i, P_F, P_{\overline{F}})$. For any general committee $W$, we can decompose its score as follows:
		
		\begin{align*}
			\MoveEqLeft \left. score_{PAV}(W) \right| \, (P_i, P_F, P_{\overline{F}}) \\
			&= \sum_{j \in N} H(|P_j \cap W|) \\
			&= H(|P_i \cap W|) + \sum_{j \in F} H(|P_j \cap W|) + \sum_{j \in \overline{F}} H(|P_j \cap W|) \\
			&= H(|P_i \cap W|) + \sum_{x=1}^{f} H(|\{c_1\} \cap W|) + \sum_{x=1}^{\frac{f'}{2}-1} H(|\{c_1, c_2, c_4\} \cap W|) \\
			&+ \sum_{x=1}^{\frac{f'}{2}+1} H(|\{c_1, c_3, c_4\} \cap W|) \\
			&= H(|P_i \cap W|) + f \cdot H(|\{c_1\} \cap W|) + \left(\frac{f'}{2} - 1\right) H(|\{c_1, c_2, c_4\} \cap W|) \\
			&+ \left(\frac{f'}{2} + 1\right) H(|\{c_1, c_3, c_4\} \cap W|)
		\end{align*}
		
		Now, since we know that $W \neq W_{-1}$ (due to the proportionality of the PAV rule), it follows that $|\{c_1\} \cap W| = 1$. We will also use the notation of $W_{-2}$ and $W_{-3}$ to denote the sub-profiles themselves, and recall that $P_i = \{c_1, c_2\}$:
		
		\begin{align*}
			\MoveEqLeft \left. score_{PAV}(W) \right| \, (P_i, P_F, P_{\overline{F}}) \\
			&= H(|\{c_1, c_2\} \cap W|) + f \cdot H(1) + \left(\frac{f'}{2} - 1\right) H(|W_{-3} \cap W|) \\
			&\quad+ \left(\frac{f'}{2} + 1\right) H(|W_{-2} \cap W|)
		\end{align*}
		
		We now evaluate this score for the different possible committees:
		\begin{align*}
			\MoveEqLeft \left. score_{PAV}(W_{-2}) \right| \, (P_i, P_F, P_{\overline{F}}) \\
			&= H(|\{c_1, c_2\} \cap W_{-2}|) + f \cdot H(1) + \left(\frac{f'}{2} - 1\right) H(|W_{-3} \cap W_{-2}|) \\
			&\quad+ \left(\frac{f'}{2} + 1\right) H(|W_{-2} \cap W_{-2}|) \\
			&= H(1) + f \cdot H(1) + \left(\frac{f'}{2} - 1\right) H(2) + \left(\frac{f'}{2} + 1\right) H(3) \\
			&= 1 + f \cdot 1 + \left(\frac{f'}{2} - 1\right) \left(1 + \frac{1}{2}\right) + \left(\frac{f'}{2} + 1\right) \left(1 + \frac{1}{2} + \frac{1}{3}\right) \\
			&= 1 + f + 1.5 \cdot \frac{f'}{2} - 1.5 + \frac{11}{6} \cdot \frac{f'}{2} + \frac{11}{6} = f + 1\frac{2}{3} f' + 1\frac{1}{3}
		\end{align*}
		
		\begin{align*}
			\MoveEqLeft \left. score_{PAV}(W_{-3}) \right| \, (P_i, P_F, P_{\overline{F}}) \\
			&= H(|\{c_1, c_2\} \cap W_{-3}|) + f \cdot H(1) + \left(\frac{f'}{2} - 1\right) H(|W_{-3} \cap W_{-3}|) \\
			&\quad+ \left(\frac{f'}{2} + 1\right) H(|W_{-2} \cap W_{-3}|) \\
			&= H(2) + f \cdot H(1) + \left(\frac{f'}{2} - 1\right) H(3) + \left(\frac{f'}{2} + 1\right) H(2) \\
			&= 1.5 + f \cdot 1 + \left(\frac{f'}{2} - 1\right) \left(1 + \frac{1}{2} + \frac{1}{3}\right) + \left(\frac{f'}{2} + 1\right) \left(1 + \frac{1}{2}\right) \\
			&= 1.5 + f + \frac{11}{6} \cdot \frac{f'}{2} - \frac{11}{6} + 1.5 \cdot \frac{f'}{2} + 1.5 = f + 1\frac{2}{3} f' + 1\frac{1}{6}
		\end{align*}

		\begin{align*}
			\MoveEqLeft \left. score_{PAV}(W_{-4}) \right| \, (P_i, P_F, P_{\overline{F}}) \\
			&= H(|\{c_1, c_2\} \cap W_{-4}|) + f \cdot H(1) + \left(\frac{f'}{2} - 1\right) H(|W_{-3} \cap W_{-4}|) \\
			&\quad+ \left(\frac{f'}{2} + 1\right) H(|W_{-2} \cap W_{-4}|) \\
			&= H(2) + f \cdot H(1) + \left(\frac{f'}{2} - 1\right) H(2) + \left(\frac{f'}{2} + 1\right) H(2) \\
			&= 1.5 + f \cdot 1 + \left(\frac{f'}{2} - 1\right) \left(1 + \frac{1}{2}\right) + \left(\frac{f'}{2} + 1\right) \left(1 + \frac{1}{2}\right) \\
			&= 1.5 + f + 1.5 \cdot \frac{f'}{2} - 1.5 + 1.5 \cdot \frac{f'}{2} + 1.5 = f + 1.5 f' + 1.5
		\end{align*}
		
		Note that, as expected, $W_{-4}$ receives a significantly lower score than the others (since the smaller coefficient of $f'$ guarantees this inequality holds for any $n \ge 6$). Furthermore, the score of $W_{-2}$ is strictly greater than the score of $W_{-3}$ (by a difference of $\frac{1}{6}$). Given these scores, the winning committee must be $W_{-2}$. That is, $\mathcal{R}(P_i, P_F, P_{\overline{F}}) = W_{-2} = \{c_1, c_3, c_4\}$.
		
		\medskip
		
		We now examine the outcome under the manipulated profile $(A_i, P_F, P_{\overline{F}})$. The initial steps of the decomposition remain identical. Recall that $A_i = \{c_2\}$; thus, for a general committee $W \neq W_{-1}$, we have:
		\begin{align*}
			\MoveEqLeft \left. score_{PAV}(W) \right| \, (A_i, P_F, P_{\overline{F}}) \\
			&= H(|A_i \cap W|) + f \cdot H(|\{c_1\} \cap W|) + \left(\frac{f'}{2} - 1\right) H(|\{c_1, c_2, c_4\} \cap W|) \\
			&\quad+ \left(\frac{f'}{2} + 1\right) H(|\{c_1, c_3, c_4\} \cap W|) \\
			&= H(|\{c_2\} \cap W|) + f \cdot H(1) + \left(\frac{f'}{2} - 1\right) H(|W_{-3} \cap W|) \\
			&\quad+ \left(\frac{f'}{2} + 1\right) H(|W_{-2} \cap W|)
		\end{align*}
		
		We now evaluate this score explicitly for the remaining possible committees, $W_{-2}$, $W_{-3}$, and $W_{-4}$:
		
		\begin{align*}
			\MoveEqLeft \left. score_{PAV}(W_{-2}) \right| \, (A_i, P_F, P_{\overline{F}}) \\
			&= H(|\{c_2\} \cap W_{-2}|) + f \cdot H(1) + \left(\frac{f'}{2} - 1\right) H(|W_{-3} \cap W_{-2}|) \\
			&\quad+ \left(\frac{f'}{2} + 1\right) H(|W_{-2} \cap W_{-2}|) \\
			&= H(0) + f \cdot H(1) + \left(\frac{f'}{2} - 1\right) H(2) + \left(\frac{f'}{2} + 1\right) H(3) \\
			&= 0 + f \cdot 1 + \left(\frac{f'}{2} - 1\right) \left(1 + \frac{1}{2}\right) + \left(\frac{f'}{2} + 1\right) \left(1 + \frac{1}{2} + \frac{1}{3}\right) \\
			&= f + 1.5 \cdot \frac{f'}{2} - 1.5 + \frac{11}{6} \cdot \frac{f'}{2} + \frac{11}{6} = f + 1\frac{2}{3} f' + \frac{1}{3}
		\end{align*}
		
		\begin{align*}
			\MoveEqLeft \left. score_{PAV}(W_{-3}) \right| \, (A_i, P_F, P_{\overline{F}}) \\
			&= H(|\{c_2\} \cap W_{-3}|) + f \cdot H(1) + \left(\frac{f'}{2} - 1\right) H(|W_{-3} \cap W_{-3}|) \\
			&\quad+ \left(\frac{f'}{2} + 1\right) H(|W_{-2} \cap W_{-3}|) \\
			&= H(1) + f \cdot H(1) + \left(\frac{f'}{2} - 1\right) H(3) + \left(\frac{f'}{2} + 1\right) H(2) \\
			&= 1 + f \cdot 1 + \left(\frac{f'}{2} - 1\right) \left(1 + \frac{1}{2} + \frac{1}{3}\right) + \left(\frac{f'}{2} + 1\right) \left(1 + \frac{1}{2}\right) \\
			&= 1 + f + \frac{11}{6} \cdot \frac{f'}{2} - \frac{11}{6} + 1.5 \cdot \frac{f'}{2} + 1.5 = f + 1\frac{2}{3} f' + \frac{2}{3}
		\end{align*}
		
		\begin{align*}
			\MoveEqLeft \left. score_{PAV}(W_{-4}) \right| \, (A_i, P_F, P_{\overline{F}}) \\
			&= H(|\{c_2\} \cap W_{-4}|) + f \cdot H(1) + \left(\frac{f'}{2} - 1\right) H(|W_{-3} \cap W_{-4}|) \\
			&\quad+ \left(\frac{f'}{2} + 1\right) H(|W_{-2} \cap W_{-4}|) \\
			&= H(1) + f \cdot H(1) + \left(\frac{f'}{2} - 1\right) H(2) + \left(\frac{f'}{2} + 1\right) H(2) \\
			&= 1 + f \cdot 1 + \left(\frac{f'}{2} - 1\right) \left(1 + \frac{1}{2}\right) + \left(\frac{f'}{2} + 1\right) \left(1 + \frac{1}{2}\right) \\
			&= 1 + f + 1.5 \cdot \frac{f'}{2} - 1.5 + 1.5 \cdot \frac{f'}{2} + 1.5 = f + 1.5 f' + 1
		\end{align*}
		
		As the explicit calculations show, $W_{-4}$ continues to receive a significantly lower score than the others (due to the smaller coefficient of $f'$). Furthermore, the score of $W_{-3}$ is strictly greater than the score of $W_{-2}$ (by a difference of $\frac{1}{3}$). Given these scores, the winning committee under the manipulated profile must be $W_{-3}$. That is, $\mathcal{R}(A_i, P_F, P_{\overline{F}}) = W_{-3} = \{c_1, c_2, c_4\}$.
		
		\medskip
		
		Recall that voter $i$'s true preferences are $P_i = \{c_1, c_2\}$. Under the truthful profile, the winning committee was $W_{-2}$, which intersects with voter $i$'s true preferences only at $\{c_1\}$. By insincerely submitting the manipulated ballot $A_i = \{c_2\}$, voter $i$ changes the outcome to $W_{-3}$, which intersects with their true preferences at $\{c_1, c_2\}$. Since $\{c_1\} \subsetneq \{c_1, c_2\}$, the new outcome provides a strict superset of the original outcome's intersection.
		
		Therefore, this explicit strict superset satisfies Condition 2 of Definition \ref{def:f-s-rasvul-dc}. Since Condition 1 also holds, this concludes the proof that the PAV rule satisfies $f$-S-RASVul-DC for $f = \left\lceil \frac{n}{k} \right\rceil$.
		\qed
	\end{proof}
	
	\bigskip
	
	\phantomsection
	\label{proof:pav-proofness}
	\noindent\textbf{Theorem~\ref{thm:pav-proofness}.} \textit{\textThmPavProofness}
	
	\begin{proof}
		
		The formal proof of this theorem follows directly as a special case of our generalized result for the $d$-RPAV rule, which we establish later in Theorem \ref{thm:drpav-proofness}. Because the proof structure for standard PAV is essentially identical to the generalized proof for $d$-RPAV, we avoid redundancy and do not present the full, duplicated proof here. Instead, we will rely on Theorem \ref{thm:drpav-proofness} and demonstrate that for $d=1$, it directly proves our theorem.
		
		First, we observe that standard PAV is exactly equivalent to $d$-RPAV for the parameter $d=1$. Recall that the $d$-RPAV rule is defined by the function $H_d(x) = \sum_{j=1}^{x} \frac{d}{j+d-1}$. By substituting $d=1$, we obtain:
		\[
		H_1(x) = \sum_{j=1}^{x} \frac{1}{j+1-1} = \sum_{j=1}^{x} \frac{1}{j} = H(x)
		\]
		Since their scoring functions perfectly coincide, the rules are identical for $d=1$. Having established this equivalence, we can derive the precise safety bound for standard PAV by substituting $d=1$ into the generalized bound of $f = \lfloor \frac{dn}{k+2d-1} \rfloor - 1$ provided by Theorem \ref{thm:drpav-proofness}:
		\[
		f = \left\lfloor \frac{1 \cdot n}{k+2(1)-1} \right\rfloor - 1 = \left\lfloor \frac{n}{k+1} \right\rfloor - 1
		\]
		Consequently, the required bound is formally established, and the proof of the theorem concludes here.
		
		\bigskip
		
		\noindent \textbf{Worst-Case Analysis.}
		
		Nevertheless, while the theoretical guarantee is fully covered, the general proof of Theorem \ref{thm:drpav-proofness} is highly intricate. To create a gradual buildup and provide better intuition, we explicitly analyze the manipulator's worst-case scenario for standard PAV. By examining how the manipulation fails in this extreme case, we aim to illustrate the underlying mechanics of the proof for the general case and provide intuition for what actually determines this specific bound.
		
		\medskip
		
		Assume by contradiction that the mechanism does not satisfy $f$-S-RASProof-DC for $f = \lfloor \frac{n}{k+1} \rfloor - 1$. We will show that the conditions for vulnerability cannot hold.
		
		By our assumption, there exists a voter $i \in N$ and a subset of known voters $F \subseteq N \setminus \{i\}$ with $|F| = \lfloor \frac{n}{k+1} \rfloor - 1$, such that for voter $i$ there exists a truthful ballot $P_i$ and a manipulated ballot $A_i \subsetneq P_i$. We will demonstrate that Condition 1 of Definition \ref{def:f-s-rasvul-dc} is violated.
		
		Since $A_i \subsetneq P_i$, there must exist at least one candidate who is in $P_i$ but not in $A_i$. Let us denote one such candidate as $c'$. Thus, we have $c' \in P_i$ and $c' \notin A_i$. Additionally, let us denote the size of the remaining set of unknown voters as $|\overline{F}| = f'$ (where $f' = n - f - 1$).
		
		\smallskip
		
		To prove that this manipulation is unsafe (violating Condition 1 of Definition \ref{def:f-s-rasvul-dc}), we must construct a profile $P_{\overline{F}}$ where dropping $c'$ strictly harms the manipulator. We will establish this by addressing the absolute hardest scenario to prove this vulnerability, i.e. the ``best-case'' scenario for the manipulator's chances of safely dropping $c'$.
		
		The worst-case scenario occurs when every known voter $j \in F$ casts the exact ballot $P_j = \{c'\}$, and furthermore, $c'$ has lexicographic priority over all other candidates (meaning $c'$ will win any tie-breaking). In this situation, the manipulating voter has certain knowledge of maximal support for $c'$ (relative to the size of $F$). Given the constraint of knowing only $f$ other ballots, this optimal state grants the manipulator the strongest possible guarantee that $c'$ might be elected even without their own approval. We will show that even under these optimal conditions for $c'$, the manipulation of dropping $c'$ is not safe.
		
		\smallskip
		
		The core intuition behind the failure of this manipulation relies on the capacity of the unknown voters ($\overline{F}$) to create a razor-thin margin between two nearly identical committees: one containing the dropped candidate $c'$, and another containing a completely unapproved outsider $c''$. Because the PAV rule operates on diminishing marginal contributions, the unknown voters can theoretically balance the scores such that the manipulator's truthful approval of $c'$ acts as the decisive weight. By strategically dropping $c'$ to free-ride on the support of $F$, the manipulator inadvertently removes this crucial support, causing the outsider $c''$ to win instead. This dynamic ensures that dropping candidates is inherently risky, as it can result in a strictly worse outcome for the manipulator.
		
		\medskip
		
		Let $c_1, \ldots, c_{k-1}, c'' \in C$ be candidates such that $c_1 \neq c', \ldots, c_{k-1} \neq c'$, and $c'' \neq c'$. We aim to select the candidates $c_1, \ldots, c_{k-1}$ such that they are included in $A_i$ as much as possible, assigning them the lowest possible indices. We choose the candidate $c''$ to be one who is not even included in $P_i$ (such a candidate must exist, since by definition, an approval ballot is a strict subset of $C$).
		
		Let $l$ denote the number of candidates from $\{c_1, \ldots, c_{k-1}\}$ that are present in $A_i$, and let $t$ denote the number of candidates from $\{c_1, \ldots, c_{k-1}\}$ that are present in $P_i$.
		
		We construct the profile of the remaining voters, $P_{\overline{F}}$, such that exactly $f \cdot k + 1$ voters from $\overline{F}$ vote for $\{c_1, \ldots, c_{k-1}, c''\}$, and the rest vote for $\{c_1, \ldots, c_{k-1}\}$. Since $\overline{F}$ contains $f' = n - f - 1$ voters, we must first verify that this chosen quantity of voters is valid (i.e., that there are enough voters in $\overline{F}$):
		
		\begin{align*}
			f \cdot k + 1 &= f(k+1) - f + 1 = \left( \left\lfloor \frac{n}{k+1} \right\rfloor - 1 \right)(k+1) - f + 1 \\
			&= (k+1)\left\lfloor \frac{n}{k+1} \right\rfloor - (k+1) - f + 1 \le n - k - 1 - f + 1 \\
			&= (n - f - 1) - k + 1 = f' - (k-1) \le f'
		\end{align*}
		
		Thus, there are indeed more than enough voters in $\overline{F}$ to construct this profile.
		
		Let $W' = \{c_1, \ldots, c_{k-1}, c'\}$ and $W'' = \{c_1, \ldots, c_{k-1}, c''\}$. We will now determine which of these committees maximizes the PAV score and becomes the winning committee in each scenario.
		
		First, we calculate the scores under the manipulated profile $(A_i, P_F, P_{\overline{F}})$:
		
		\begin{align*}
			\MoveEqLeft \left. score_{PAV}(W'') \right| \, (A_i, P_F, P_{\overline{F}}) \\
			&= H(|A_i \cap W''|) + \sum_{j \in F} H(|\{c'\} \cap W''|) + \sum_{j \in \overline{F}} H(|P_j \cap W''|) \\
			&= H(l) + f \cdot H(0) + \big(f' - (f \cdot k + 1)\big) \cdot H(k-1) + (f \cdot k + 1) \cdot H(k) \\
			&= H(l) + 0 + f' \cdot H(k-1) + (f \cdot k + 1) \cdot \big(H(k) - H(k-1)\big) \\
			&= H(l) + f' \cdot H(k-1) + (f \cdot k + 1) \cdot \frac{1}{k} \\
			&= H(l) + f' \cdot H(k-1) + f \cdot k \cdot \frac{1}{k} + \frac{1}{k} \\
			&= H(l) + f' \cdot H(k-1) + f + \frac{1}{k}
		\end{align*}
		
		\begin{align*}
			\MoveEqLeft \left. score_{PAV}(W') \right| \, (A_i, P_F, P_{\overline{F}}) \\
			&= H(|A_i \cap W'|) + \sum_{j \in F} H(|\{c'\} \cap W'|) + \sum_{j \in \overline{F}} H(|P_j \cap W'|) \\
			&= H(l) + f \cdot H(1) + f' \cdot H(k-1) \\
			&= H(l) + f + f' \cdot H(k-1)
		\end{align*}
		
		Thus, $W''$ achieves a strictly greater score than $W'$ (by a difference of $\frac{1}{k}$). Therefore, $W''$ defeats $W'$ under the manipulated profile.
		
		Next, we examine the committees under the truthful profile $(P_i, P_F, P_{\overline{F}})$:
		
		\begin{align*}
			\MoveEqLeft \left. score_{PAV}(W'') \right| \, (P_i, P_F, P_{\overline{F}}) \\
			&= H(|P_i \cap W''|) + \sum_{j \in F} H(|\{c'\} \cap W''|) + \sum_{j \in \overline{F}} H(|P_j \cap W''|) \\
			&= H(t) + f \cdot H(0) + \big(f' - (f \cdot k + 1)\big) \cdot H(k-1) + (f \cdot k + 1) \cdot H(k) \\
			&= H(t) + 0 + f' \cdot H(k-1) + (f \cdot k + 1) \cdot \big(H(k) - H(k-1)\big) \\
			&= H(t) + f' \cdot H(k-1) + (f \cdot k + 1) \cdot \frac{1}{k} \\	
			&= H(t) + f' \cdot H(k-1) + f \cdot k \cdot \frac{1}{k} + \frac{1}{k} \\
			&= H(t) + f' \cdot H(k-1) + f + \frac{1}{k}
		\end{align*}
		
		\begin{align*}
			\MoveEqLeft \left. score_{PAV}(W') \right| \, (P_i, P_F, P_{\overline{F}}) \\
			&= H(|P_i \cap W'|) + \sum_{j \in F} H(|\{c'\} \cap W'|) + \sum_{j \in \overline{F}} H(|P_j \cap W'|) \\
			&= H(t+1) + f \cdot H(1) + f' \cdot H(k-1) \\
			&= H(t) + \frac{1}{t+1} + f + f' \cdot H(k-1)
		\end{align*}
		
		Under the truthful profile, the score of $W'$ exceeds the score of $W''$ by $\frac{1}{t+1} - \frac{1}{k}$. Recall that $t$ was defined as the number of candidates from $\{c_1, \ldots, c_{k-1}\}$ present in $P_i$. Thus, it is clear that $t \le k-1$, which implies $t+1 \le k$, and consequently $\frac{1}{t+1} \ge \frac{1}{k}$. Since the difference $\frac{1}{t+1} - \frac{1}{k} \ge 0$ is non-negative, $W'$ defeats $W''$ in this scenario. Recall that we assumed $c'$ has lexicographic priority over all other candidates; therefore, even if this inequality is tight (i.e., the difference is exactly zero), $W'$ will still win the tie-breaking against $W''$.
		
		\medskip
		
		It remains to be shown that under both the manipulated profile $(A_i, P_F, P_{\overline{F}})$ and the truthful profile $(P_i, P_F, P_{\overline{F}})$, no other committee can defeat $W''$ and $W'$, respectively.
		
		Suppose there exists such a committee, denoted by $W^*$, and we will examine various cases for it. First, we consider the case where it is composed strictly of the candidates $\{c_1, \ldots, c_{k-1}, c'', c'\}$. There are exactly $k+1$ candidates here, and thus exactly one candidate must be excluded to form a committee of size $k$. The possibilities of excluding $c''$ or $c'$ have already been analyzed (these correspond to $W''$ and $W'$, respectively). Therefore, we examine the possibility of excluding one candidate from $\{c_1, \ldots, c_{k-1}\}$.
		
		\begin{align*}
			\MoveEqLeft \left. score_{PAV}(W^*) \right| \, (A_i, P_F, P_{\overline{F}}) \\
			&= H(|A_i \cap W^*|) + \sum_{j \in F} H(|\{c'\} \cap W^*|) + \sum_{j \in \overline{F}} H(|P_j \cap W^*|) \\
			&\overset{(1)}{\le} H(l) + f \cdot H(1) + \big(f' - (f \cdot k + 1)\big) \cdot H(k-2) + (f \cdot k + 1) \cdot H(k-1) \\
			&= H(l) + f + f' \cdot H(k-2) + (f \cdot k + 1) \cdot \big(H(k-1) - H(k-2)\big) \\
			&= H(l) + f + f' \cdot H(k-1) - f' \cdot \frac{1}{k-1} + (f \cdot k + 1) \cdot \frac{1}{k-1} \\
			&= H(l) + f + f' \cdot H(k-1) - \big(f' - (f \cdot k + 1)\big) \cdot \frac{1}{k-1} \\
			&\overset{(2)}{\le} H(l) + f + f' \cdot H(k-1) \\
			&< H(l) + f' \cdot H(k-1) + f + \frac{1}{k}
		\end{align*}
		
		\smallskip
		\noindent \footnotesize
		$^{(1)}$ Holds because if the excluded candidate is in $A_i$, we obtain $H(l-1)$; if not, we obtain $H(l)$. \\
		$^{(2)}$ Holds because we previously showed that $f \cdot k + 1 \le f'$, which implies $f' - (f \cdot k + 1) \ge 0$.
		\normalsize
		\medskip
		
		This demonstrates that the score obtained for $W^*$ is strictly less than the score $W''$ received under the manipulated profile $(A_i, P_F, P_{\overline{F}})$. Therefore, $W''$ defeats such committees.
		
		\medskip
		
		Next, we examine this for the truthful profile $(P_i, P_F, P_{\overline{F}})$:
		\begin{align*}
			\MoveEqLeft \left. score_{PAV}(W^*) \right| \, (P_i, P_F, P_{\overline{F}}) \\
			&= H(|P_i \cap W^*|) + \sum_{j \in F} H(|\{c'\} \cap W^*|) + \sum_{j \in \overline{F}} H(|P_j \cap W^*|) \\
			&\overset{(1)}{\le} H(t) + f \cdot H(1) + \big(f' - (f \cdot k + 1)\big) \cdot H(k-2) + (f \cdot k + 1) \cdot H(k-1) \\
			&= H(t) + f + f' \cdot H(k-2) + (f \cdot k + 1) \cdot \big(H(k-1) - H(k-2)\big) \\
			&= H(t) + f + f' \cdot H(k-1) - f' \cdot \frac{1}{k-1} + (f \cdot k + 1) \cdot \frac{1}{k-1} \\
			&= H(t) + f + f' \cdot H(k-1) - \big(f' - (f \cdot k + 1)\big) \cdot \frac{1}{k-1} \\
			&\overset{(2)}{\le} H(t) + f + f' \cdot H(k-1) \\
			&< H(t) + f' \cdot H(k-1) + f + \frac{1}{t+1}
		\end{align*}
		
		\smallskip
		\noindent \footnotesize
		The algebraic transitions and the justifications for $(1)$ and $(2)$ are equivalent to the previous case.
		\normalsize
		\medskip
		
		This shows that the score obtained for $W^*$ is strictly less than the score $W'$ received under the truthful profile $(P_i, P_F, P_{\overline{F}})$. Hence, $W'$ defeats such committees as well.
		\medskip
		
		We now consider what happens if $W^*$ includes candidates that are not from the set $\{c_1, \ldots, c_{k-1}, c'', c'\}$. Such candidates are neither present in the ballots of $F$ nor in the ballots of $\overline{F}$; thus, they will definitively contribute nothing to the score from these subsets of voters. 
		
		If these added candidates are in $A_i$, it means that all candidates from $\{c_1, \ldots, c_{k-1}\}$ are already in $A_i$ as well (due to our deliberate selection of this set). In this case, adding one of these candidates to the committee is equivalent to swapping one of the original candidates for them. While such an addition can contribute at most $\frac{1}{k}$ to the score of voter $i$, it will simultaneously cause a loss of at least $f' \cdot \frac{1}{k}$ from the score of $\overline{F}$ (since no candidate from the original set contributes less than $\frac{1}{k}$ to $\overline{F}$, and we have removed a contributing candidate to add a non-contributing one). Consequently, such committees will definitively fail to win. 
		
		If the added candidates are not in $A_i$, and also not in $P_i$, then no voter casts a ballot for them whatsoever. Thus, they cannot contribute to the score at all, but only detract from it (by displacing a candidate who might have contributed). 
		
		The only remaining case is adding candidates who are in $P_i$ but not in $A_i$. Under the manipulated profile $(A_i, P_F, P_{\overline{F}})$, they obviously contribute nothing, and therefore $W''$ will still defeat such committees. Under the truthful profile $(P_i, P_F, P_{\overline{F}})$, we can say these candidates are equivalent to $c'$, but they will be unable to displace $c'$ from the winning committee. This is because $c'$ receives additional votes from the entire set $F$ (and even in the extreme case where $F$ is empty, $c'$ holds lexicographic priority). Therefore, even if we cannot definitively state that $W'$ is the unique winning committee for $(P_i, P_F, P_{\overline{F}})$ in this specific sub-case, we can state with absolute certainty that $c'$ will be a member of the winning committee.
		
		Thus, we have established that $\mathcal{R}(A_i, P_F, P_{\overline{F}}) = W''$, and $c' \in \mathcal{R}(P_i, P_F, P_{\overline{F}})$. Since $c' \notin W''$, it immediately follows that $c' \notin \mathcal{R}(A_i, P_F, P_{\overline{F}})$. 
		
		Recall that we initially defined $c' \in P_i$. Therefore, we have necessarily found a profile where:
		\[
		\mathcal{R}(A_i, P_F, P_{\overline{F}}) \cap P_i \not\supseteq \mathcal{R}(P_i, P_F, P_{\overline{F}}) \cap P_i
		\]
		
		This means the manipulation is not safe (it violates Condition~1 of Definition~\ref{def:f-s-rasvul-dc}). We have thus shown that the theorem holds for this worst-case scenario.
		\qed
	\end{proof}

	\section{Proofs for Section~\ref{sec:drpav-rule}(\nameref{sec:drpav-rule})}
	\label{app:proofs-drpav}
	
	\phantomsection
	\label{proof:drpav-alpha-jr}
	\noindent\textbf{Theorem~\ref{thm:drpav-alpha-jr}.} \textit{\textThmDrpavAlphaJR}
	
	\begin{proof}
		
		Assume towards a contradiction that the $d$-RPAV rule violates $\alpha$-JR where $\alpha = d$. This means there exists a profile $P$ and a group of voters $N^* \subseteq N$ of size $|N^*| \ge \left\lceil \frac{dn}{k} \right\rceil$ who mutually approve at least one common candidate (i.e., $\bigcap_{i \in N^*} P_i \neq \emptyset$), yet no voter in $N^*$ has an approved candidate in the winning committee (i.e., $P_i \cap \mathcal{R}(P) = \emptyset$ for all $i \in N^*$).
		
		Let $W = \mathcal{R}(P)$. For each $w \in W$, we define its \emph{marginal contribution} as the difference between $score_{d-RPAV}(W)$ and $score_{d-RPAV}(W \setminus \{w\})$ (Note that $score_{d\text{-}RPAV}$ is well-defined for any subset of $C$, regardless of size; in particular, $W \setminus \{w\}$, which has size $k-1$, need not constitute a valid committee, yet its score is still well-defined). Let $s_W$ be the sum of the marginal contributions of all candidates in $W$.
		\medskip
		
		Consider a voter $i \in N \setminus N^*$ and examine their contribution to $s_W$. Let $r_i = |P_i \cap W|$ be the number of candidates that voter $i$ approved who were elected to the winning committee. If $r_i = 0$, it is clear that voter $i$ contributes nothing to $s_W$. Otherwise, if $r_i > 0$, then each candidate in $P_i \cap W$ contributes $\frac{d}{r_i+d-1}$ to $s_W$ through voter $i$'s ballot, and therefore all candidates in $P_i \cap W$ combined contribute to $s_W$ via voter $i$'s vote as follows:
		\[
		\sum_{w \in P_i \cap W} \frac{d}{r_i+d-1} = |P_i \cap W| \cdot \frac{d}{r_i+d-1} = r_i \cdot \frac{d}{r_i+d-1} \stackrel{(1)}{\le} r_i \cdot \frac{d}{r_i} = d
		\]
		
		\noindent \footnotesize
		{(1) $r_i + d - 1 \ge r_i$, as $d\ge1$.}
		\normalsize
		\medskip
		
		Thus, we have shown that when $r_i = 0$ the contribution is $0$, and when $r_i > 0$ the contribution is at most $d$; hence, in any case, the contribution of a single voter to $s_W$ does not exceed $d$.
		\medskip
		
		Since, by our assumption, no voter in $N^*$ has an approved candidate in $W$, the committee's score from the voters in $N^*$ is zero. Thus, the marginal contributions to $s_W$ are derived exclusively from the group $N \setminus N^*$, and we can bound $s_W$ as follows:
		
		\[
		s_W \le \sum_{i \in N \setminus N^*} d = |N \setminus N^*| \cdot d = (n - |N^*|) \cdot d < nd
		\]
		
		Since we have shown that $s_W < nd$, by the pigeonhole principle there must exist a candidate $w \in W$ whose marginal contribution is strictly less than:
		\[
		\frac{1}{|W|} \cdot nd = \frac{nd}{k} \le \left\lceil \frac{dn}{k} \right\rceil
		\]
		
		Now, if we remove this candidate $w$ from the winning committee $W$ and add in their place some candidate from $\bigcap_{i \in N^*} P_i$ (we know such a candidate necessarily exists by our initial assumption), we will gain at least the score from all the voters in $N^*$ (who previously contributed nothing to the score). Their total contribution will be:

		\begin{align*}
			\MoveEqLeft \sum_{i \in N^*} H_d(1)  = |N^*|  \cdot H_d(1)  \ge  \left\lceil \frac{dn}{k} \right\rceil \cdot H_d(1) \\
			&=\left\lceil \frac{dn}{k} \right\rceil \cdot \sum_{j=1}^{1} \frac{d}{j+d-1} = \left\lceil \frac{dn}{k} \right\rceil \cdot \frac{d}{1+d-1} = \left\lceil \frac{dn}{k} \right\rceil \cdot 1 = \left\lceil \frac{dn}{k} \right\rceil
		\end{align*}

		We have shown that removing candidate $w$ from $W$ and adding a candidate from $\bigcap_{i \in N^*} P_i$ decreases the score by strictly less than $\left\lceil \frac{dn}{k} \right\rceil$ and increases it by at least $\left\lceil \frac{dn}{k} \right\rceil$. Therefore, this swap strictly increases the overall score of $W$, which contradicts the assumption that $W$ is the winning committee (as it must maximize the score).
		\qed
	\end{proof}

	\bigskip
	
	\phantomsection
	\label{proof:drpav-proofness}
	\noindent\textbf{Theorem~\ref{thm:drpav-proofness}.} \textit{\textThmDrpavProofness}	
	
	\begin{proof}
		
		We prove the theorem by contradiction, showing that Condition~1 of Definition~\ref{def:f-s-rasvul-dc} must be violated. The proof proceeds in five steps. In \textbf{Step 1}, we construct two ``witness'' committees, $W'$ and $W''$, together with a tailored voting profile for $\overline{F}$. In \textbf{Step 2}, we compute the $d$-RPAV scores of $W'$ and $W''$ under both the manipulated and the truthful profiles, and choose the free parameters $x'$ and $x''$ so that each committee wins under its intended profile. In \textbf{Step 3}, we verify that this choice of $x'$ and $x''$ is feasible. In \textbf{Step 4}, we rule out the ``hybrid'' committee $W^*$ as a winner under the manipulated profile. In \textbf{Step 5}, we rule out every remaining committee, thereby establishing that $c'$ is excluded from the winning committee under manipulation yet included under truthful reporting, the desired contradiction.
		
		\medskip
		
		\noindent\textbf{ Step 1: Construction.}

		Assume by contradiction that the mechanism does not satisfy $f$-S-RASProof-DC for $f = \left\lfloor \frac{dn}{k+2d-1}\right\rfloor - 1 $. We will show that the conditions for vulnerability cannot hold.
		
		By our assumption, there exists a voter $i \in N$ and a subset of voters $F \subseteq N \setminus \{i\}$ with $|F| = \left\lfloor \frac{dn}{k+2d-1}\right\rfloor - 1 $, such that for voter $i$ there exists a truthful ballot $P_i$ and a manipulated ballot $A_i \subsetneq P_i$. We will demonstrate that Condition 1 of Definition \ref{def:f-s-rasvul-dc} is violated.
		
		Since $A_i \subsetneq P_i$, there must exist at least one candidate who is in $P_i$ but not in $A_i$. Let us denote one such candidate as $c'$. Thus, we have $c' \in P_i$ and $c' \notin A_i$. Additionally, let us denote the size of the remaining set of voters as $|\overline{F}| = f'$ (where $f' = n - f - 1$).
		
		\medskip
		
		Since finding a single profile for $\overline{F}$ is sufficient to prove the manipulation is unsafe, we first evaluate the hypothetical election outcome for the subset $F \cup \{i\}$ to serve as a reference point for explicitly tailoring the ballots of $\overline{F}$ to this specific instance. Thus, we will "run" the $d$-RPAV mechanism as if the only voters present were $F \cup \{i\}$, where voter $i$'s ballot is $A_i$. Let us denote the winning committee in this execution by $W_{F, A_i}$, and the score this committee received by $w_{F, A_i}$.
		
		Let $G_i = P_i \setminus A_i$. This is essentially the set from which we selected $c'$. We will select $c'$ specifically, with the sole strict requirement being $c' \in G_i$ in all cases. Additionally, we will define various candidates using the notations $c''$ and $c_1, \ldots, c_{k-1}$. For candidate $c''$, we will require that $c'' \notin G_i$. This is necessarily possible because approval ballots must be strict subsets of the candidate set $C$ and cannot be empty, ensuring there are always candidates not in $G_i$. Furthermore, since $|W_{F, A_i}| = k$ and by definition $1 \le k < m$, there are necessarily some candidates inside $W_{F, A_i}$ and some candidates outside of $W_{F, A_i}$.
		
		We divide the proof into three distinct cases:
		\begin{itemize}
			\item \textbf{Case 1:} $W_{F, A_i} \subseteq G_i$.
			\item \textbf{Case 2:} $W_{F, A_i} \cap G_i = \emptyset$.
			\item \textbf{Case 3:} All other scenarios. That is, the winning committee $W_{F, A_i}$ contains both candidates who are in $G_i$ and candidates who are not in $G_i$. Simply put, $W_{F, A_i} \cap G_i \neq \emptyset$ and $W_{F, A_i} \setminus G_i \neq \emptyset$.
		\end{itemize}
		
		We further divide Case 3 into two sub-cases:
		\begin{itemize}
			\item \textbf{Case 3.1:} $(C \setminus W_{F, A_i}) \cap G_i = \emptyset$.
			\item \textbf{Case 3.2:} $(C \setminus W_{F, A_i}) \cap G_i \neq \emptyset$.
		\end{itemize}
		
		Note that in both Case 1 and Case 3.1, it holds that $W_{F, A_i} \cap G_i \neq \emptyset$ and $(C \setminus W_{F, A_i}) \setminus G_i \neq \emptyset$. Similarly, note that in both Case 2 and Case 3.2, it holds that $W_{F, A_i} \setminus G_i \neq \emptyset$ and $(C \setminus W_{F, A_i}) \cap G_i \neq \emptyset$.

		First, we define the choices for Case 1 and Case 3.1, where we know that $W_{F, A_i} \cap G_i \neq \emptyset$. We examine each candidate in $W_{F, A_i} \cap G_i$ to determine what happens if we remove only that candidate from the committee $W_{F, A_i}$, specifically by how much the score of $W_{F, A_i}$ would decrease without them under the profile $(P_F, P_i)$. The candidate whose removal decreases the score by the least amount will be denoted as $c'$. In the event of a tie among several candidates, we choose the one that is lexicographically "weakest". Note that since we selected from $W_{F, A_i} \cap G_i$, it is clear that $c' \in G_i$, as required.
		
		Since $W_{F, A_i}$ contains exactly $k$ candidates, we denote the remaining candidates, $W_{F, A_i} \setminus \{c'\}$, as $c_1, \ldots, c_{k-1}$.
		
		Next, we examine the remaining candidates who are neither in $W_{F, A_i}$ nor in $G_i$, to determine which candidate would contribute the most to the $d$-RPAV score if added to the committee $\{c_1, \ldots, c_{k-1}\}$ when our voters are $F \cup \{i\}$ and voter $i$'s ballot is $A_i$. We denote this candidate as $c''$. In other words, $c''$ is the candidate from $(C \setminus W_{F, A_i}) \setminus G_i$ whose addition to $W_{F, A_i} \setminus \{c'\}$ contributes the most to the committee's score under the profile $(P_F, A_i)$. In the event of a tie, we select the one that is lexicographically "strongest". Since in these cases it holds that $(C \setminus W_{F, A_i}) \setminus G_i \neq \emptyset$, such a candidate clearly exists, and according to the set from which they were chosen, $c'' \notin G_i$.
		
		We now handle Case 2 and Case 3.2, where we know that $W_{F, A_i} \setminus G_i \neq \emptyset$. We examine each candidate in $W_{F, A_i} \setminus G_i$ to determine what happens if we remove only that candidate from the committee $W_{F, A_i}$, observing by how much the score of $W_{F, A_i}$ would decrease without them under the profile $(P_F, A_i)$. The candidate whose removal decreases the score by the least amount will be denoted as $c''$. If there is a tie, we choose the one that is lexicographically "weakest". Note that since we selected from $W_{F, A_i} \setminus G_i$, it is clear that $c'' \notin G_i$.
		
		Since $W_{F, A_i}$ contains exactly $k$ candidates, we denote the remaining candidates, $W_{F, A_i} \setminus \{c''\}$, as $c_1, \ldots, c_{k-1}$.
		
		Next, we examine the candidates in $G_i$ who are not in $W_{F, A_i}$, to determine which candidate would contribute the most to the $d$-RPAV score if added to the committee $\{c_1, \ldots, c_{k-1}\}$ when our voters are $F \cup \{i\}$ and voter $i$'s ballot is now $P_i$. We denote this candidate as $c'$. In other words, $c'$ is the candidate from $(C \setminus W_{F, A_i}) \cap G_i$ whose addition to $W_{F, A_i} \setminus \{c''\}$ contributes the most to the committee's score under the profile $(P_F, P_i)$. In the event of a tie, we select the one that is lexicographically "strongest". Since in these cases it holds that $(C \setminus W_{F, A_i}) \cap G_i \neq \emptyset$, such a candidate clearly exists, and according to the set from which they were chosen, $c' \in G_i$, as required.
		
		Thus, across all cases, we have explicitly defined which candidate from $G_i$ is chosen as $c'$, and we have selected the candidates $c_1, \ldots, c_{k-1}$ as well as the candidate $c''$. From here, we will proceed uniformly with our selection for the various cases.
		
		\medskip
		
		Let $W' = \{c_1, \ldots, c_{k-1}, c'\}$ and $W'' = \{c_1, \ldots, c_{k-1}, c''\}$. Furthermore, let us denote $C_{K-1} = \{c_1, \ldots, c_{k-1}\}$.
		
		Let $s_F$ denote the score contributed by the subset $F$ to the candidates in $C_{K-1}$; that is, $s_F \coloneqq \sum_{j \in F} H_d(|P_j \cap C_{K-1}|)$.
		
		Let $s_{c'}$ denote the marginal score that $c'$ contributes to the score of $W'$ from the subset $F$. That is, $s_{c'} \coloneqq \sum_{j \in F} H_d(|P_j \cap W'|) - s_F$. Similarly, let $s_{c''}$ denote the marginal score that $c''$ contributes to the score of $W''$ from the subset $F$, meaning $s_{c''} \coloneqq \sum_{j \in F} H_d(|P_j \cap W''|) - s_F$.
		
		Let $l$ denote the number of candidates in $C_{K-1}$ that are present in $A_i$, and let $t$ denote the number of candidates in $C_{K-1}$ that are present in $P_i$. Naturally, it holds that $0 \le l \le k-1$ and $0 \le t \le k-1$.
		
		We wish to denote the marginal contribution of candidate $c''$ to the score of $W''$ from voter $i$. Since we do not know a priori whether $c'' \in A_i$ or $c'' \notin P_i$, we denote the marginal contribution of $c''$ under $P_i$ as $p''$, and the marginal contribution of $c''$ under $A_i$ as $a''$. Specifically, if $c'' \in A_i$, then $p'' = \frac{d}{t+d}$ and $a'' = \frac{d}{l+d}$. Otherwise (i.e., if $c'' \notin P_i$), $p'' = 0$ and $a'' = 0$.
		
		We choose the profile $P_{\overline{F}}$ such that all voters in $\overline{F}$ vote for $\{c_1, \ldots, c_{k-1}\}$. Among them, $x'$ voters will additionally vote for $c'$, and $x''$ voters will additionally vote for $c''$, where we aim to maximize the overlap between these $x'$ and $x''$ voters (this overlap will help us in \hyperlink{proof:drpav:step4}{Step 4}, where we need to rule out any committee containing both $c'$ and $c''$; forcing these candidates to share voters minimizes their combined score due to diminishing marginal returns). Explicitly, $\min(x', x'')$ voters from $\overline{F}$ will vote for $\{c_1, \ldots, c_{k-1}, c', c''\}$. Additionally, $\max(x'-x'', 0)$ voters will vote for $\{c_1, \ldots, c_{k-1}, c'\}$, $\max(x''-x', 0)$ voters will vote for $\{c_1, \ldots, c_{k-1}, c''\}$, and the remainder, which is $f' - \max(x', x'')$ voters, will vote strictly for $\{c_1, \ldots, c_{k-1}\}$. We will define the exact values of $x'$ and $x''$ later, but note that they must satisfy $0 \le x' \le f'$ and $0 \le x'' \le f'$.
		
		\noindent\textbf{Step 2: Scoring $W'$ and $W''$, and the tie-break argument.}
		
		Having fixed the construction, we now determine the $d$-RPAV scores of $W'$ and $W''$ under both the manipulated and the truthful profiles, and choose $x'$ and $x''$ so that each committee defeats the other under its intended profile.
		
		\noindent\textbf{Scores under both profiles.}
		
		We now examine the scores of the potential winning committees in each scenario.

		Under the manipulated profile $(A_i, P_F, P_{\overline{F}})$:
		\begin{align*}
			\MoveEqLeft \left. score_{d-RPAV}(W'') \right| \, (A_i, P_F, P_{\overline{F}}) \\
			&= H_d(|A_i \cap W''|) + \sum_{j \in F} H_d(|P_j \cap W''|) + \sum_{j \in \overline{F}} H_d(|P_j \cap W''|) \\
			&= H_d(l) + a'' + s_F + s_{c''} + (f' - x'') \cdot H_d(k-1) + x'' \cdot H_d(k) \\
			&= H_d(l) + a'' + s_F + s_{c''} + f' \cdot H_d(k-1) + x'' \cdot \big(H_d(k) - H_d(k-1)\big) \\
			&= H_d(l) + a'' + s_F + s_{c''} + f' \cdot H_d(k-1) + x'' \cdot \frac{d}{k+d-1}
		\end{align*}
		
		\begin{align*}
			\MoveEqLeft \left. score_{d-RPAV}(W') \right| \, (A_i, P_F, P_{\overline{F}}) \\
			&= H_d(|A_i \cap W'|) + \sum_{j \in F} H_d(|P_j \cap W'|) + \sum_{j \in \overline{F}} H_d(|P_j \cap W'|) \\
			&= H_d(l) + s_F + s_{c'} + (f' - x') \cdot H_d(k-1) + x' \cdot H_d(k) \\
			&= H_d(l) + s_F + s_{c'} + f' \cdot H_d(k-1) + x' \cdot \big(H_d(k) - H_d(k-1)\big) \\
			&= H_d(l) + s_F + s_{c'} + f' \cdot H_d(k-1) + x' \cdot \frac{d}{k+d-1}
		\end{align*}
		
		Thus, we obtain a gap of $(x'' - x') \cdot \frac{d}{k+d-1} + a'' + s_{c''} - s_{c'}$ in favor of $W''$ over $W'$.
		
		Now, we evaluate the committees under the truthful profile $(P_i, P_F, P_{\overline{F}})$:
		\begin{align*}
			\MoveEqLeft \left. score_{d-RPAV}(W'') \right| \, (P_i, P_F, P_{\overline{F}}) \\
			&= H_d(|P_i \cap W''|) + \sum_{j \in F} H_d(|P_j \cap W''|) + \sum_{j \in \overline{F}} H_d(|P_j \cap W''|) \\
			&= H_d(t) + p'' + s_F + s_{c''} + (f' - x'') \cdot H_d(k-1) + x'' \cdot H_d(k) \\
			&= H_d(t) + p'' + s_F + s_{c''} + f' \cdot H_d(k-1) + x'' \cdot \big(H_d(k) - H_d(k-1)\big) \\
			&= H_d(t) + p'' + s_F + s_{c''} + f' \cdot H_d(k-1) + x'' \cdot \frac{d}{k+d-1}
		\end{align*}
		
		\begin{align*}
			\MoveEqLeft \left. score_{d-RPAV}(W') \right| \, (P_i, P_F, P_{\overline{F}}) \\
			&= H_d(|P_i \cap W'|) + \sum_{j \in F} H_d(|P_j \cap W'|) + \sum_{j \in \overline{F}} H_d(|P_j \cap W'|) \\
			&= H_d(t+1) + s_F + s_{c'} + (f' - x') \cdot H_d(k-1) + x' \cdot H_d(k) \\
			&= H_d(t) + \frac{d}{t+d} + s_F + s_{c'} + f' \cdot H_d(k-1) + x' \cdot \big(H_d(k) - H_d(k-1)\big) \\
			&= H_d(t) + \frac{d}{t+d} + s_F + s_{c'} + f' \cdot H_d(k-1) + x' \cdot \frac{d}{k+d-1}
		\end{align*}
		
		Here, we obtain a gap of $\frac{d}{t+d} - p'' + s_{c'} - s_{c''} - (x'' - x') \cdot \frac{d}{k+d-1}$ in favor of $W'$ over $W''$.
		
		\noindent\textbf{Choosing $x'-x''$ to force the intended outcome.}
		
		We will now divide the proof into cases and determine the values of $x'$ and $x''$ for each.
		
		First, we handle the scenario where $W'$ is lexicographically stronger than $W''$. In this situation, the gap $(x'' - x') \cdot \frac{d}{k+d-1} + a'' + s_{c''} - s_{c'}$ must be strictly positive, so that $W''$ defeats $W'$ under the manipulation. And meanwhile, the gap $\frac{d}{t+d} - p'' + s_{c'} - s_{c''} - (x'' - x') \cdot \frac{d}{k+d-1}$ must be non-negative, so that $W'$ defeats $W''$ without the manipulation. To satisfy this, we choose $x'$ and $x''$ such that $x'' - x' = \left\lfloor \frac{k+d-1}{d}(s_{c'} - s_{c''} - a'') \right\rfloor + 1$, and we will show that this holds.
		
		\begin{align*}
			\MoveEqLeft (x'' - x') \cdot \frac{d}{k+d-1} + a'' + s_{c''} - s_{c'} \\
			&= \left( \left\lfloor \frac{k+d-1}{d}(s_{c'} - s_{c''} - a'') \right\rfloor + 1 \right) \cdot \frac{d}{k+d-1} + a'' + s_{c''} - s_{c'} \\
			&> \left( \frac{k+d-1}{d}(s_{c'} - s_{c''} - a'') - 1 + 1 \right) \cdot \frac{d}{k+d-1} + a'' + s_{c''} - s_{c'} \\
			&= (s_{c'} - s_{c''} - a'') + a'' + s_{c''} - s_{c'} \\
			&= 0
		\end{align*}
		
		And similarly for the second gap:
		\begin{align*}
			\MoveEqLeft \frac{d}{t+d} - p'' + s_{c'} - s_{c''} - (x'' - x') \cdot \frac{d}{k+d-1} \\
			&= \frac{d}{t+d} - p'' + s_{c'} - s_{c''} \\
			&\quad- \left( \left\lfloor \frac{k+d-1}{d}(s_{c'} - s_{c''} - a'') \right\rfloor + 1 \right) \cdot \frac{d}{k+d-1} \\
			&\ge \frac{d}{t+d} - p'' + s_{c'} - s_{c''} \\
			&\quad- \left( \frac{k+d-1}{d}(s_{c'} - s_{c''} - a'') + 1 \right) \cdot \frac{d}{k+d-1} \\
			&= \frac{d}{t+d} - p'' + s_{c'} - s_{c''} - (s_{c'} - s_{c''} - a'') - \frac{d}{k+d-1} \\
			&= \frac{d}{t+d} - p'' + a'' - \frac{d}{k+d-1}
		\end{align*}
		
		We have shown that the expression required to be strictly positive is indeed positive. For the side required to be non-negative, we obtained the expression $\frac{d}{t+d} - p'' + a'' - \frac{d}{k+d-1}$. Thus, it remains to be shown that this expression itself is non-negative, which we will demonstrate shortly.
		
		Next, we handle the scenario where $W''$ is lexicographically stronger than $W'$. In this situation, the gap $(x'' - x') \cdot \frac{d}{k+d-1} + a'' + s_{c''} - s_{c'}$ must be non-negative, so that $W''$ defeats $W'$ under the manipulation. And meanwhile, the gap $\frac{d}{t+d} - p'' + s_{c'} - s_{c''} - (x'' - x') \cdot \frac{d}{k+d-1}$ must be strictly positive, so that $W'$ defeats $W''$ without the manipulation. To satisfy this, we choose $x'$ and $x''$ such that $x'' - x' = \left\lceil \frac{k+d-1}{d}(s_{c'} - s_{c''} - a'') \right\rceil$, and we will show that this holds.
		
		\begin{align*}
			\MoveEqLeft (x'' - x') \cdot \frac{d}{k+d-1} + a'' + s_{c''} - s_{c'} \\
			&= \left\lceil \frac{k+d-1}{d}(s_{c'} - s_{c''} - a'') \right\rceil \cdot \frac{d}{k+d-1} + a'' + s_{c''} - s_{c'} \\
			&\ge \frac{k+d-1}{d}(s_{c'} - s_{c''} - a'') \cdot \frac{d}{k+d-1} + a'' + s_{c''} - s_{c'} \\
			&= (s_{c'} - s_{c''} - a'') + a'' + s_{c''} - s_{c'} \\
			&= 0
		\end{align*}
		
		And similarly for the second gap:
		\begin{align*}
			\MoveEqLeft \frac{d}{t+d} - p'' + s_{c'} - s_{c''} - (x'' - x') \cdot \frac{d}{k+d-1} \\
			&= \frac{d}{t+d} - p'' + s_{c'} - s_{c''} \\
			&\quad- \left\lceil \frac{k+d-1}{d}(s_{c'} - s_{c''} - a'') \right\rceil \cdot \frac{d}{k+d-1} \\
			&> \frac{d}{t+d} - p'' + s_{c'} - s_{c''} \\
			&\quad- \left( \frac{k+d-1}{d}(s_{c'} - s_{c''} - a'') + 1 \right) \cdot \frac{d}{k+d-1} \\
			&= \frac{d}{t+d} - p'' + s_{c'} - s_{c''} - (s_{c'} - s_{c''} - a'') - \frac{d}{k+d-1} \\
			&= \frac{d}{t+d} - p'' + a'' - \frac{d}{k+d-1}
		\end{align*}
		
		We have shown here as well that the expression required to be non-negative is indeed non-negative. For the strictly positive side, the algebraic path yielded a strict inequality directly. Thus, it remains to show that $\frac{d}{t+d} - p'' + a'' - \frac{d}{k+d-1}$ is a non-negative expression, exactly as before. We will demonstrate this now.
		
		If $c'' \in A_i$, then $p'' = \frac{d}{t+d}$ and $a'' = \frac{d}{l+d}$. Therefore:
		\[
		\frac{d}{t+d} - p'' + a'' - \frac{d}{k+d-1} = \frac{d}{t+d} - \frac{d}{t+d} + \frac{d}{l+d} - \frac{d}{k+d-1} = \frac{d}{l+d} - \frac{d}{k+d-1}
		\]
		Since $l \le k-1$, it implies $l+d \le k+d-1$, and thus $\frac{d}{l+d} \ge \frac{d}{k+d-1}$. Therefore, $\frac{d}{l+d} - \frac{d}{k+d-1} \ge 0$. We have shown that this is indeed a non-negative expression.
		
		If $c'' \notin P_i$, then $p'' = 0$ and $a'' = 0$. Therefore:
		\[
		\frac{d}{t+d} - p'' + a'' - \frac{d}{k+d-1} = \frac{d}{t+d} - 0 + 0 - \frac{d}{k+d-1} = \frac{d}{t+d} - \frac{d}{k+d-1}
		\]
		Since $t \le k-1$, it implies $t+d \le k+d-1$, and thus $\frac{d}{t+d} \ge \frac{d}{k+d-1}$. Therefore, $\frac{d}{t+d} - \frac{d}{k+d-1} \ge 0$. We have shown that this is indeed a non-negative expression.
		
		\noindent\textbf{Step 3: Feasibility of $x'$ and $x''$.}
		
		We now confirm that the values of $x'$ and $x''$ chosen in Step 2 are actually realizable, i.e., that they satisfy $0 \le x' \le f'$ and $0 \le x'' \le f'$.		
		
		When we defined $x'$ and $x''$ earlier, we stated that they must satisfy $0 \le x' \le f'$ and $0 \le x'' \le f'$. For it to be possible to assign values that satisfy this condition, we must demonstrate that $-f' \le x'' - x' \le f'$.
		
		Since the definition of the difference $(x'' - x')$ utilizes the values $s_{c'}$ and $s_{c''}$, we first examine the bounds on their sizes:
		\begin{align*}
			s_{c'} &= \sum_{j \in F} H_d(|P_j \cap W'|) - s_F \\
			&= \sum_{j \in F} H_d(|P_j \cap W'|) - \sum_{j \in F} H_d(|P_j \cap C_{K-1}|) \\
			&= \sum_{j \in F} \big( H_d(|P_j \cap W'|) - H_d(|P_j \cap C_{K-1}|) \big) \\
			&\le \sum_{j \in F} 1 = f \cdot 1 = f
		\end{align*}
		Furthermore, it is clear that $s_{c'} \ge 0$. We similarly show this for $s_{c''}$:
		\begin{align*}
			s_{c''} &= \sum_{j \in F} H_d(|P_j \cap W''|) - s_F \\
			&= \sum_{j \in F} H_d(|P_j \cap W''|) - \sum_{j \in F} H_d(|P_j \cap C_{K-1}|) \\
			&= \sum_{j \in F} \big( H_d(|P_j \cap W''|) - H_d(|P_j \cap C_{K-1}|) \big) \\
			&\le \sum_{j \in F} 1 = f \cdot 1 = f
		\end{align*}
		Here as well, it is clear that $s_{c''} \ge 0$. 
		
		We now show that $-f' \le x'' - x' \le f'$, starting with the case where $x'' - x' = \left\lfloor \frac{k+d-1}{d}(s_{c'} - s_{c''} - a'') \right\rfloor + 1$. Recall that by definition $f' = n - f - 1$, and in our case $f = \left\lfloor \frac{dn}{k+2d-1}\right\rfloor - 1 $.
		
		\begin{align*}
			x'' - x' &= \left\lfloor \frac{k+d-1}{d}(s_{c'} - s_{c''} - a'') \right\rfloor + 1 \\
			&\ge \left\lfloor \frac{k+d-1}{d}(-s_{c''} - a'') \right\rfloor + 1 \\
			&\ge \left\lfloor \frac{k+d-1}{d}(-s_{c''} - 1) \right\rfloor + 1 \\
			&\ge \left\lfloor \frac{k+d-1}{d}(-f - 1) \right\rfloor + 1 \\
			&= \left\lfloor -f \cdot \frac{k+d-1}{d} - \frac{k+d-1}{d} \cdot 1 \right\rfloor + 1 \\		
			&= \left\lfloor -f \cdot \frac{k+d-1}{d} -f \cdot \frac{d}{d} +f \cdot \frac{d}{d} -  \frac{k+d-1}{d} \right\rfloor + 1 \\		
			&= \left\lfloor -f \cdot \frac{k+2d-1}{d} +f \cdot 1 -  \frac{k+d-1}{d} \right\rfloor + 1 \\		
			&= \left\lfloor -\left( \left\lfloor \frac{dn}{k+2d-1} \right\rfloor - 1 \right) \cdot \frac{k+2d-1}{d} + f -  \frac{k+d-1}{d} \right\rfloor + 1 \\
			&= \left\lfloor -\left\lfloor \frac{dn}{k+2d-1} \right\rfloor \cdot \frac{k+2d-1}{d} + \frac{k+2d-1}{d} + f -  \frac{k+d-1}{d} \right\rfloor + 1 \\
			&\ge \left\lfloor -\left( \frac{dn}{k+2d-1} \right) \cdot \frac{k+2d-1}{d} + \frac{k+2d-1}{d} + f -  \frac{k+d-1}{d} \right\rfloor + 1 \\
			&= \left\lfloor - \frac{dn}{d} + \frac{k}{d} + \frac{2d}{d} - \frac{1}{d} + f -  \frac{k}{d} - \frac{d}{d} + \frac{1}{d} \right\rfloor + 1 \\
			&= \left\lfloor - n + f + \frac{d}{d} \right\rfloor + 1 \\
			&= \left\lfloor - n + f + 1 \right\rfloor + 1 \\
			&= ( - n + f + 1 ) + 1 \\
			&= -( n - f - 1) + 1 \\
			&= -f' + 1 \\
			&> -f'
		\end{align*}
		
		And for the other side:
		\begin{align*}
			x'' - x' &= \left\lfloor \frac{k+d-1}{d}(s_{c'} - s_{c''} - a'') \right\rfloor + 1 \\
			&\le \left\lfloor \frac{k+d-1}{d}(s_{c'}) \right\rfloor + 1 \\		
			&\le \left\lfloor \frac{k+d-1}{d}(f) \right\rfloor + 1 \\
			&= \left\lfloor \frac{k+d-1}{d} \cdot f + \frac{d}{d} \cdot f - \frac{d}{d} \cdot f \right\rfloor + 1 \\
			&= \left\lfloor \frac{k+2d-1}{d} \cdot f - 1 \cdot f \right\rfloor + 1 \\
			&= \left\lfloor \frac{k+2d-1}{d} \cdot \left( \left\lfloor \frac{dn}{k+2d-1} \right\rfloor - 1 \right) - f \right\rfloor + 1 \\
			&= \left\lfloor \frac{k+2d-1}{d} \cdot \left\lfloor \frac{dn}{k+2d-1} \right\rfloor - \frac{k+2d-1}{d} - f \right\rfloor + 1 \\
			&\le \left\lfloor \frac{k+2d-1}{d} \cdot \left( \frac{dn}{k+2d-1} \right) - \frac{k+2d-1}{d} - f \right\rfloor + 1 \\
			&= \left\lfloor \frac{dn}{d} - \frac{k}{d} - \frac{2d}{d} + \frac{1}{d} - f \right\rfloor + 1 \\
			&= \left\lfloor n - \frac{k}{d} - 2 + \frac{1}{d} - f \right\rfloor + 1 \\
			&= \left\lfloor (n - f - 1) - 1 - \frac{k-1}{d} \right\rfloor + 1 \\
			&= (n - f - 1) - 1 + \left\lfloor - \frac{k-1}{d} \right\rfloor + 1 \\		
			&= f' + \left\lfloor \frac{1-k}{d} \right\rfloor \\
			&\le f' + 0 \\
			&= f'
		\end{align*}
		
		We also show this for the case where $x'' - x' = \left\lceil \frac{k+d-1}{d}(s_{c'} - s_{c''} - a'') \right\rceil$:
		\begin{align*}
			x'' - x' &= \left\lceil \frac{k+d-1}{d}(s_{c'} - s_{c''} - a'') \right\rceil \\
			&\ge \left\lceil \frac{k+d-1}{d}(-s_{c''} - a'') \right\rceil \\
			&\ge \left\lceil \frac{k+d-1}{d}(-s_{c''} - 1) \right\rceil \\		
			&\ge \left\lceil \frac{k+d-1}{d}(-f - 1) \right\rceil \\
			&= \left\lceil -\frac{k+d-1}{d} \cdot f - \frac{k+d-1}{d} \right\rceil \\
			&= \left\lceil -\frac{k+d-1}{d} \cdot f -\frac{d}{d} \cdot f +\frac{d}{d} \cdot f - \frac{k+d-1}{d} \right\rceil \\
			&= \left\lceil -\frac{k+2d-1}{d} \cdot f + 1 \cdot f - \frac{k+d-1}{d} \right\rceil \\
			&= \left\lceil -\frac{k+2d-1}{d} \cdot \left( \left\lfloor \frac{dn}{k+2d-1} \right\rfloor - 1 \right) + f - \frac{k+d-1}{d} \right\rceil \\
			&= \left\lceil -\frac{k+2d-1}{d} \cdot \left\lfloor \frac{dn}{k+2d-1} \right\rfloor + \frac{k+2d-1}{d} + f - \frac{k+d-1}{d} \right\rceil \\
			&\ge \left\lceil -\frac{k+2d-1}{d} \cdot \left( \frac{dn}{k+2d-1} \right) + \frac{k+2d-1}{d} + f - \frac{k+d-1}{d} \right\rceil \\
			&= \left\lceil - \frac{dn}{d} + \frac{k}{d} + \frac{2d}{d} - \frac{1}{d} + f - \frac{k}{d} - \frac{d}{d} + \frac{1}{d} \right\rceil  \\
			&= \left\lceil - n + f + \frac{d}{d} \right\rceil  \\		
			&= \left\lceil - n + f + 1 \right\rceil  \\		
			&= - n + f + 1 \\		
			&= - (n - f - 1) \\
			&= - f' \\
		\end{align*}
		
		And for the other side:
		\begin{align*}
			x'' - x' &= \left\lceil \frac{k+d-1}{d}(s_{c'} - s_{c''} - a'') \right\rceil \\
			&\le \left\lceil \frac{k+d-1}{d} \cdot s_{c'} \right\rceil \\
			&\le \left\lceil \frac{k+d-1}{d} \cdot f \right\rceil \\
			&= \left\lceil \frac{k+d-1}{d} \cdot f + \frac{d}{d} \cdot f - \frac{d}{d} \cdot f \right\rceil \\
			&= \left\lceil \frac{k+2d-1}{d} \cdot f - 1 \cdot f \right\rceil \\
			&= \left\lceil \frac{k+2d-1}{d} \cdot \left( \left\lfloor \frac{dn}{k+2d-1} \right\rfloor - 1 \right) - f \right\rceil \\
			&= \left\lceil \frac{k+2d-1}{d} \cdot \left\lfloor \frac{dn}{k+2d-1} \right\rfloor - \frac{k+2d-1}{d} - f \right\rceil \\
			&\le \left\lceil \frac{k+2d-1}{d} \cdot \left( \frac{dn}{k+2d-1} \right) - \frac{k+2d-1}{d} - f \right\rceil \\
			&= \left\lceil \frac{dn}{d} - \frac{k}{d} - \frac{2d}{d} + \frac{1}{d} - f \right\rceil \\
			&= \left\lceil n - \frac{k}{d} - 2 + \frac{1}{d} - f \right\rceil \\
			&= \left\lceil (n - f - 1) - 1 - \frac{k-1}{d} \right\rceil \\
			&= (n - f - 1) + \left\lceil  - 1 - \frac{k-1}{d} \right\rceil \\
			&= (n - f - 1) - \left\lfloor 1 + \frac{k-1}{d} \right\rfloor \\
			&\le (n - f - 1) - 1 \\
			&< (n - f - 1) \\
			&= f'
		\end{align*}
		
		Thus, we have shown that there exist $x'$ and $x''$ that satisfy the requirements. If, according to our definition of $x'' - x'$, it turns out that $x'' - x' = f'$, we will be forced to choose $x'' = f'$ and $x' = 0$. If it turns out that $x'' - x' = -f'$, we will be forced to choose $x'' = 0$ and $x' = f'$. In any other case, we will have flexibility. However, we specifically require that if $x'' - x' = 0$, we explicitly choose $x'' = 0$ and $x' = 0$.
		
		\medskip
		
		\noindent\textbf{\hypertarget{proof:drpav:step4}{Step 4}: Ruling out the hybrid committee $W^*$.}
		
		Having secured the two witness committees $W'$ and $W''$ relative to each other, we now begin eliminating all other candidate committees. We first handle the ``hybrid'' committee $W^*$, obtained by swapping a single member of $\{c_1,\ldots,c_{k-1}\}$ for the unused candidate from $\{c', c''\}$.
		
		We will need to ensure that $W^*$ cannot defeat $W''$ under the manipulated profile $(A_i, P_F, P_{\overline{F}})$, nor defeat $W'$ under the truthful profile $(P_i, P_F, P_{\overline{F}})$. We will examine various configurations for $W^*$ (depending on the excluded candidate). Formally, we consider a committee composed strictly from the candidates $\{c_1, \ldots, c_{k-1}, c'', c'\}$. There are exactly $k+1$ candidates here, and therefore exactly one candidate must be excluded to form a committee of size $k$. The possibilities of excluding $c''$ or $c'$ have already been analyzed in $W'$ and $W''$; thus, we examine the possibility of excluding one candidate from $\{c_1, \ldots, c_{k-1}\}$. Let us denote this excluded candidate by $c^{-}$. We now analyze the score of $W^*$ under the different voting profiles.
		
		\begin{align*}
			\MoveEqLeft \left. score_{d-RPAV}(W^*) \right| \, (A_i, P_F, P_{\overline{F}}) \\
			&= H_d(|A_i \cap W^*|) + \sum_{j \in F} H_d(|W^* \cap P_j|) + \sum_{j \in \overline{F}} H_d(|W^* \cap P_j|) \\
			\MoveEqLeft \left. score_{d-RPAV}(W^*) \right| \, (P_i, P_F, P_{\overline{F}}) \\
			&= H_d(|P_i \cap W^*|) + \sum_{j \in F} H_d(|W^* \cap P_j|) + \sum_{j \in \overline{F}} H_d(|W^* \cap P_j|)
		\end{align*}
		
		We will decompose the various components separately. The contribution from the voters in $\overline{F}$ is:
		\begin{align*}
			\MoveEqLeft \sum_{j \in \overline{F}} H_d(|W^* \cap P_j|) \\
			&= \min(x', x'') \cdot H_d(k) + \big(\max(x', x'') - \min(x', x'')\big) \cdot H_d(k-1) \\
			&\quad+ \big(f' - \max(x', x'')\big) \cdot H_d(k-2) \\
			&= \min(x', x'') \cdot \big(H_d(k) - H_d(k-1)\big) \\
			&\quad+ \max(x', x'') \cdot \big(H_d(k-1) - H_d(k-2)\big) \\
			&\quad + f' \cdot H_d(k-2) + f' \cdot \frac{d}{k+d-2} - f' \cdot \frac{d}{k+d-2} \\
			&= \min(x', x'') \cdot \frac{d}{k+d-1} + \max(x', x'') \cdot \frac{d}{k+d-2} \\
			&\quad+ f' \cdot \left(H_d(k-2) + \frac{d}{k+d-2}\right) - f' \cdot \frac{d}{k+d-2} \\
			&= \min(x', x'') \cdot \frac{d}{k+d-1} + f' \cdot H_d(k-1) - \frac{d}{k+d-2}\big(f' - \max(x', x'')\big)
		\end{align*}
		
		Recall that the contribution of $\overline{F}$ to the score of $W'$ was $f' \cdot H_d(k-1) + x' \cdot \frac{d}{k+d-1}$, and the contribution of $\overline{F}$ to the score of $W''$ was $f' \cdot H_d(k-1) + x'' \cdot \frac{d}{k+d-1}$.
		
		If $x'' > x'$, then the contribution of $\overline{F}$ to the score of $W^*$ is $x' \cdot \frac{d}{k+d-1} + f' \cdot H_d(k-1) - \frac{d}{k+d-2}(f' - x'')$. Since $x'' \le f'$, we can state that this contribution is at most $x' \cdot \frac{d}{k+d-1} + f' \cdot H_d(k-1)$. This means it is less than or equal to the contribution of $\overline{F}$ to the score of $W'$, and it is strictly less than the contribution of $\overline{F}$ to the score of $W''$ (since $x'' > x'$).
		
		If $x' > x''$, then the contribution of $\overline{F}$ to the score of $W^*$ is $x'' \cdot \frac{d}{k+d-1} + f' \cdot H_d(k-1) - \frac{d}{k+d-2}(f' - x')$. Since $x' \le f'$, we can state that this contribution is at most $x'' \cdot \frac{d}{k+d-1} + f' \cdot H_d(k-1)$. This means it is less than or equal to the contribution of $\overline{F}$ to the score of $W''$, and it is strictly less than the contribution of $\overline{F}$ to the score of $W'$ (since $x' > x''$).
		
		If $x'' = x'$ (meaning $x'' - x' = 0$), in this case we previously defined that we explicitly choose $x'' = 0$ and $x' = 0$. Therefore, we can state that the contribution of $\overline{F}$ to the score of $W^*$ is $f' \cdot H_d(k-1) - \frac{d}{k+d-2} \cdot f'$. This means this contribution is strictly less than both the contribution of $\overline{F}$ to the score of $W''$ and the contribution of $\overline{F}$ to the score of $W'$.
		
		Note that if $W^*$ wins under the profile $(P_i, P_F, P_{\overline{F}})$, this does not interfere with our proof, because $c' \in W^*$; meaning that $c'$ is in the winning committee under truthful reporting, which is exactly what we want to prove.

		We now examine the side of $(A_i, P_F)$. We divide the analysis according to the cases with which we originally selected the candidates, and we will prove that in all of them, $W''$ defeats $W^*$ under the profile $(A_i, P_F, P_{\overline{F}})$.
		
		In Case 1 and Case 3.1, assume by contradiction that under the profile $(A_i, P_F, P_{\overline{F}})$, the committee $W^*$ defeats the committee $W''$. Since we have already proven that under this same profile $W''$ defeats $W'$, then under this profile $W^*$ must also defeat $W'$. Since we have already seen that from $P_{\overline{F}}$ the score of $W'$ is greater than or equal to that of $W^*$, this implies that even under the profile $(A_i, P_F)$, $W^*$ defeats $W'$. But since in these cases we can state that $W_{F, A_i} = W'$, and since $W_{F, A_i}$ is the winning committee in $(A_i, P_F)$, it is clear by definition that $W'$ will defeat any other committee under the profile $(A_i, P_F)$, including our $W^*$. We have reached a contradiction, and therefore we have proven that $W''$ will also defeat $W^*$ under the profile $(A_i, P_F, P_{\overline{F}})$.
		
		In Case 2 and Case 3.2, it is even simpler. In these cases, we can state that $W_{F, A_i} = W''$, and since $W_{F, A_i}$ is the winning committee in $(A_i, P_F)$, it is clear that $W''$ will receive from $(A_i, P_F)$ a score greater than or equal to any other committee, including our $W^*$. Therefore, since we have seen that in all cases $W''$ will receive from $P_{\overline{F}}$ a score that is at least greater than or equal to the score $W^*$ receives from it, and since in the event of a tie $W_{F, A_i}$ is lexicographically stronger, we can state that $W''$ will also defeat $W^*$ under the voting profile $(A_i, P_F, P_{\overline{F}})$.
		
		\bigskip
		
		\noindent\textbf{Step 5: Ruling out all remaining committees.}
		
		We now examine all other committees. Any additional committee we have not examined contains at least one candidate who is not in the set $\{c_1, \ldots, c_{k-1}, c'', c'\}$. We want to prove that no such committee that does not contain $c'$ will win under the profile $(P_i, P_F, P_{\overline{F}})$, and that no such committee that contains $c'$ will win under the profile $(A_i, P_F, P_{\overline{F}})$.
		
		To prove this, we will start from a specific committee that is desirable in that profile, and we will swap candidates one by one, where each time we remove a candidate from the committee from the set $\{c_1, \ldots, c_{k-1}, c'', c'\}$, and we add a candidate to the committee, whom we denote by $c^{+} \in C \setminus \{c_1, \ldots, c_{k-1}, c'', c'\}$. We will prove that in each such step the score does not increase, and therefore ultimately this theoretical committee will not be able to win.
		
		\noindent\textbf{A reusable swap fact.}
		
		Before analyzing the specific profiles, we establish the following lemma, which will be used several times.
		
		\begin{lemma}\label{lemma:swap_fact}
			For any committee containing a candidate $c^- \in \{c_1, \ldots, c_{k-1}\}$, swapping $c^-$ for a candidate $c^+ \in C \setminus \{c_1, \ldots, c_{k-1}, c'', c'\}$ results in a committee with a strictly lower $d$-RPAV score.
		\end{lemma}
		
		\begin{proof}
			A candidate from the set $\{c_1, \ldots, c_{k-1}\}$ receives votes from all of $\overline{F}$, while the candidate we add, $c^{+}$, from the set $C \setminus \{c_1, \ldots, c_{k-1}, c'', c'\}$, receives no votes at all from $\overline{F}$. Therefore, this loss will decrease the score by at least $\frac{1}{k}$ from each of the voters in $\overline{F}$. This means a total minimal reduction of $\frac{d}{k+d-1} \cdot f'$. On the other hand, the candidate $c^{+}$ may receive a score of at most $1$ from everyone in the set $F \cup \{i\}$ (regardless of whether their ballot is $(P_i, P_F)$ or $(A_i, P_F)$), and the removal of the candidate from the set $\{c_1, \ldots, c_{k-1}\}$ might not decrease the score at all, meaning the maximal score addition from the set $F \cup \{i\}$ would be $1 \cdot (f+1)$. This means the maximal addition we could get from this kind of swap would be $(f+1) - \frac{d}{k+d-1} \cdot f'$. We will show that it will not increase the score at all:
			
			\begin{align*}
				\MoveEqLeft (f+1) - \frac{d}{k+d-1} \cdot f' \\
				&= \frac{1}{k+d-1} \cdot \left((k+d-1) \cdot (f+1) - d \cdot f'\right) \\
				&= \frac{1}{k+d-1} \cdot \left((k+d-1) \cdot (f+1) - d \cdot (n - f - 1)\right) \\
				&= \frac{1}{k+d-1} \cdot \left(k \cdot f + d \cdot f - f + k + d - 1 - d \cdot n + d \cdot f + d\right) \\
				&= \frac{1}{k+d-1} \cdot \left(k \cdot f + 2d \cdot f - f + k + 2d -1 - d \cdot n\right) \\
				&= \frac{1}{k+d-1} \cdot \left((k + 2d - 1) \cdot f + k + 2d - 1 - d \cdot n \right) \\
				&= \frac{1}{k+d-1} \cdot \left((k + 2d - 1) \cdot \left( \left\lfloor \frac{dn}{k+2d-1} \right\rfloor - 1 \right) + k + 2d - 1 - d \cdot n \right) \\
				&\begin{aligned}= \frac{1}{k+d-1} \cdot  &\bigg((k + 2d - 1) \cdot \left\lfloor \frac{dn}{k+2d-1} \right\rfloor - (k + 2d - 1) \\
					&\quad+ k + 2d - 1 - d \cdot n \bigg)\end{aligned} \\
				&= \frac{1}{k+d-1} \cdot \left((k + 2d - 1) \cdot \left\lfloor \frac{dn}{k+2d-1} \right\rfloor - d \cdot n \right) \\
				&\le \frac{1}{k+d-1} \cdot \left((k + 2d - 1) \cdot \left( \frac{dn}{k+2d-1} \right) - d \cdot n \right) \\
				&= \frac{1}{k+d-1} \cdot \left(dn - d \cdot n \right) \\
				&= \frac{1}{k+d-1} \cdot (0) \\
				&= 0
			\end{align*}

			This means we have shown that swapping a candidate from the set $\{c_1, \ldots, \allowbreak c_{k-1}\}$ in exchange for adding a candidate from the set $C \setminus \{c_1, \ldots, c_{k-1}, c'', c'\}$ yields, in the worst-case scenario, a score addition of $0$. 
			In practice, this swap will even yield a strictly negative score addition in all cases, because for the score of $c^{+}$ from the set $F \cup \{i\}$ to truly reach $1 \cdot (f+1)$, $c^{+}$ would need to be the only candidate that the entire set voted for, at least under the profile $(A_i, P_F)$; in such a situation, it is clear $c^+$ would be selected for $W_{F, A_i}$, and therefore cannot be from the set $C \setminus \{c_1, \ldots, c_{k-1}, c'', c'\}$. Thus, such a swap will strictly decrease the score.  \hfill $\blacksquare$
		\end{proof}

		\medskip
		
		\noindent\textbf{No $c'$-free committee wins under the truthful profile.}
		
		We now begin to prove that no committee that does not contain $c'$ will win under the profile $(P_i, P_F, P_{\overline{F}})$. 
		
		If there is such a committee, we divide it into two cases: either it contains $c''$ or it does not contain $c''$. 
		
		If it contains $c''$, we start from the committee $W''$, which we have already seen is defeated by $W'$ under the profile $(P_i, P_F, P_{\overline{F}})$, and we swap candidates one by one. 
		Since the committee we are examining does not contain $c'$ and does contain $c''$, any removal of a candidate from $W''$ will be of a candidate from the set $\{c_1, \ldots, c_{k-1}\}$, and any candidate addition will be an addition of a candidate from the set $C \setminus \{c_1, \ldots, c_{k-1}, c'', c'\}$. 
		We have seen in Lemma \ref{lemma:swap_fact} that such a swap strictly decreases the score. 
		
		If it does not contain $c''$, the first step will be the removal of $c'$ itself from $W'$ and replacing it with another candidate $c^{+} \in C \setminus \{c_1, \ldots, c_{k-1}, c'', c'\}$. 
		The subsequent steps will be swapping a candidate from $\{c_1, \ldots, c_{k-1}\}$ with a candidate from the set $C \setminus \{c_1, \ldots, c_{k-1}, c'', c'\}$, from which we strictly lose, as stated above (Lemma \ref{lemma:swap_fact}); therefore, we will examine only the first step. 
		
		It is clear that from $P_{\overline{F}}$, $W'$ can only benefit, because there are no votes for $c^{+}$ in $P_{\overline{F}}$, and there might be votes for $c'$. 
		We will show that even from $(P_i, P_F)$, $W'$ will not lose in the various cases.

		Let us denote the committee after the first swap by $W^{**}$, and note that it is essentially $\{c_1, \ldots, c_{k-1}, c^{+}\}$. 
		
		In Case 1 and Case 3.1, $W_{F, A_i} = W'$, meaning it wins for $(A_i, P_F)$. 
		$W^{**}$ has a certain score in $(A_i, P_F)$, and it is not greater than the score of $W'$. 
		If we move to $(P_i, P_F)$, the portion of the score each of the committees receives from $P_F$ will not change. 
		As for the portion from $P_i$, the score of $W^{**}$ can increase from a minimum of $H_d(l)$ to a maximum of $H_d(t+1)$. 
		Since the score of $W'$ already increases from $H_d(l)$ to $H_d(t+1)$ due to the exact same change, the score of $W^{**}$ will not increase by more than that of $W'$. 
		Therefore, $W'$ will defeat $W^{**}$ also under the profile $(P_i, P_F)$ in these cases. 
		
		In Case 2 and Case 3.2, if $c^{+} \in G_i$, then if $W^{**}$ were to defeat $W'$ under the profile $(P_i, P_F)$, then $c^{+}$ should have been selected in these cases instead of $c'$ exactly according to the mechanism by which we chose $c'$. 
		If $c^{+} \notin G_i$, then it is indeed possible that $W^{**}$ defeats $W'$ under the profile $(P_i, P_F)$. 
		Therefore, we will step back slightly in the proof, and we will prove that in this case $W'$ defeats $W^{**}$ under the profile $(P_i, P_F, P_{\overline{F}})$. 
		We will need significantly more detail for this.

		\medskip
		
		Recall the notations $s_{c'}$ and $s_{c''}$. Similarly, we define $s_{c^{+}}$ as the marginal score that $c^{+}$ contributes to the score of $W^{**}$ from the subset $F$, meaning $s_{c^{+}} \coloneqq \sum_{j \in F} H_d(|P_j \cap W^{**}|) - s_F$.
		
		Recall the notations $p''$ and $a''$. Similarly, we define $p^{+}$ and $a^{+}$ to be the marginal contribution of candidate $c^{+}$ from voter $i$ to the score of $W^{**}$. Specifically, if $c^{+} \in A_i$, then $p^{+} = \frac{d}{t+d}$ and $a^{+} = \frac{d}{l+d}$; otherwise (i.e., if $c^{+} \notin P_i$), $p^{+} = 0$ and $a^{+} = 0$.
		
		Note that:
		\begin{align*}
			\MoveEqLeft \left. score_{d-RPAV}(W'') \right| \, (A_i, P_F) \\
			&= H_d(|A_i \cap W''|) + \sum_{j \in F} H_d(|W'' \cap P_j|) = H_d(l) + a'' + s_F + s_{c''}
		\end{align*}

		And that:
		\begin{align*}
			\MoveEqLeft \left. score_{d-RPAV}(W^{**}) \right| \, (A_i, P_F) \\
			&= H_d(|A_i \cap W^{**}|) + \sum_{j \in F} H_d(|W^{**} \cap P_j|) = H_d(l) + a^{+} + s_F + s_{c^{+}}	
		\end{align*}	
		
		Since in these cases $W_{F, A_i} = W''$, then $W''$ defeats any committee under the profile $(A_i, P_F)$, and in particular the committee $W^{**}$. Therefore, $H_d(l) + a'' + s_F + s_{c''} \ge H_d(l) + a^{+} + s_F + s_{c^{+}}$, which implies that $a'' + s_{c''} \ge a^{+} + s_{c^{+}}$, meaning that $s_{c^{+}} \le s_{c''} + a'' - a^{+}$. Furthermore, if the inequality is tight, then $W''$ is lexicographically preferred over $W^{**}$.
		
		Recall that the score of $W'$ under the profile $(P_i, P_F, P_{\overline{F}})$ resulted in:
		\begin{align*}
			\MoveEqLeft \left. score_{d-RPAV}(W') \right| \, (P_i, P_F, P_{\overline{F}}) \\
			&= H_d(t) + \frac{d}{t+d} + s_F + s_{c'} + f' \cdot H_d(k-1) + x' \cdot \frac{d}{k+d-1} 
		\end{align*}
		And we calculate the score of $W^{**}$ under the profile $(P_i, P_F, P_{\overline{F}})$:
		\begin{align*}
			\MoveEqLeft \left. score_{d-RPAV}(W^{**}) \right| \, (P_i, P_F, P_{\overline{F}}) \\
			&= H_d(|P_i \cap W^{**}|) + \sum_{j \in F} H_d(|W^{**} \cap P_j|) + \sum_{j \in \overline{F}} H_d(|W^{**} \cap P_j|) \\
			&= H_d(t) + p^{+} + s_F + s_{c^{+}} + f' \cdot H_d(k-1)
		\end{align*}
		
		This means the gap between $W'$ and $W^{**}$ under the profile $(P_i, P_F, P_{\overline{F}})$ is $\frac{d}{t+d} - p^{+} + s_{c'} - s_{c^{+}} + x' \cdot \frac{d}{k+d-1}$ in favor of $W'$.
		
		Recall that we defined that in the situation where $W'$ is lexicographically stronger than $W''$, then $(x'' - x') = \left\lfloor \frac{k+d-1}{d}(s_{c'} - s_{c''} - a'') \right\rfloor + 1$, and in the situation where $W''$ is lexicographically stronger than $W'$, then $(x'' - x') = \left\lceil \frac{k+d-1}{d}(s_{c'} - s_{c''} - a'') \right\rceil$.
		This means that in the situation where $W'$ is lexicographically stronger than $W''$, we can state $x' = x'' - \left\lfloor \frac{k+d-1}{d}(s_{c'} - s_{c''} - a'') \right\rfloor - 1$, and in the situation where $W''$ is lexicographically stronger than $W'$, we can state $x' = x'' - \left\lceil \frac{k+d-1}{d}(s_{c'} - s_{c''} - a'') \right\rceil$.
		
		We examine the gap in the situation where $W'$ is lexicographically stronger than $W''$:
		\begin{align*}
			\MoveEqLeft \frac{d}{t+d} - p^{+} + s_{c'} - s_{c^{+}} + x' \cdot \frac{d}{k+d-1} \\
			&= \frac{d}{t+d} - p^{+} + s_{c'} - s_{c^{+}} \\
			&\quad+ \left(x'' - \left\lfloor \frac{k+d-1}{d}(s_{c'} - s_{c''} - a'') \right\rfloor - 1\right) \cdot \frac{d}{k+d-1} \\
			&\ge \frac{d}{t+d} - p^{+} + s_{c'} - s_{c^{+}} \\
			&\quad+ \left(x'' - \frac{k+d-1}{d}(s_{c'} - s_{c''} - a'') - 1\right) \cdot \frac{d}{k+d-1} \\
			&= \frac{d}{t+d} - p^{+} + s_{c'} - s_{c^{+}} + x'' \cdot \frac{d}{k+d-1} - (s_{c'} - s_{c''} - a'') - \frac{d}{k+d-1} \\
			&= \frac{d}{t+d} - p^{+} + s_{c'} - s_{c^{+}} + x'' \cdot \frac{d}{k+d-1} - s_{c'} + s_{c''} + a'' - \frac{d}{k+d-1} \\
			&= \frac{d}{t+d} - p^{+} - s_{c^{+}} + x'' \cdot \frac{d}{k+d-1} + s_{c''} + a'' - \frac{d}{k+d-1} \\
			&\ge \frac{d}{t+d} - p^{+} - (s_{c''} + a'' - a^{+}) + x'' \cdot \frac{d}{k+d-1} + s_{c''} + a'' - \frac{d}{k+d-1} \\
			&= \frac{d}{t+d} - p^{+} - s_{c''} - a'' + a^{+} + x'' \cdot \frac{d}{k+d-1} + s_{c''} + a'' - \frac{d}{k+d-1} \\
			&= \frac{d}{t+d} - p^{+} + a^{+} + x'' \cdot \frac{d}{k+d-1} - \frac{d}{k+d-1} \\
			&\ge \frac{d}{t+d} - p^{+} + a^{+} - \frac{d}{k+d-1}
		\end{align*}
		
		If $c^{+} \notin P_i$, then the expression becomes $\frac{d}{t+d} - p^{+} + a^{+} - \frac{d}{k+d-1} = \frac{d}{t+d} - 0 + 0 - \frac{d}{k+d-1} = \frac{d}{t+d} - \frac{d}{k+d-1}$. Since $t \le k-1$, it follows that $t+d \le k+d-1$, and thus $\frac{d}{t+d} \ge \frac{d}{k+d-1}$. Therefore, $\frac{d}{t+d} - \frac{d}{k+d-1} \ge 0$, meaning we obtained a non-negative expression.
		
		Otherwise, if $c^+ \in P_i$, then since $c^+ \notin G_i$ in the current case (meaning $c^+ \notin P_i \setminus A_i$), it must be that $c^+ \in A_i$. In this case, the expression becomes $\frac{d}{t+d} - p^{+} + a^{+} - \frac{d}{k+d-1} = \frac{d}{t+d} - \frac{d}{t+d} + \frac{d}{l+d} - \frac{d}{k+d-1} = \frac{d}{l+d} - \frac{d}{k+d-1}$. Since it is known that $l \le k-1$, it follows that $l+d \le k+d-1$, and thus $\frac{d}{l+d} \ge \frac{d}{k+d-1}$. Therefore, $\frac{d}{l+d} - \frac{d}{k+d-1} \ge 0$, meaning we obtained a non-negative expression here as well.
		
		Note that during the calculation we used the inequality $s_{c^{+}} \le s_{c''} + a'' - a^{+}$, and we know that if it holds tightly, then $W''$ is lexicographically preferred over $W^{**}$. Now, if it does not hold tightly, we can replace the greater-than-or-equal sign in our notation with a strictly greater-than sign, and therefore we obtained a strictly positive gap in favor of $W'$ over $W^{**}$, so $W'$ will defeat $W^{**}$ also under the truthful profile $(P_i, P_F, P_{\overline{F}})$. If it does hold tightly, then $W''$ is lexicographically preferred over $W^{**}$, and we are in a situation where $W'$ is lexicographically stronger than $W''$; thus, we can state that $W'$ is lexicographically preferred over $W^{**}$. Therefore, since the gap is non-negative, in this case too we can state that $W'$ will defeat $W^{**}$ under the truthful profile $(P_i, P_F, P_{\overline{F}})$.
		
		We now examine the gap in the situation where $W''$ is lexicographically stronger than $W'$:
		\begin{align*}
			&\frac{d}{t+d} - p^{+} + s_{c'} - s_{c^{+}} + x' \cdot \frac{d}{k+d-1} \\
			&= \frac{d}{t+d} - p^{+} + s_{c'} - s_{c^{+}} + \big(x'' - \left\lceil \frac{k+d-1}{d}(s_{c'} - s_{c''} - a'') \right\rceil\big) \cdot \frac{d}{k+d-1} \\
			&> \frac{d}{t+d} - p^{+} + s_{c'} - s_{c^{+}} + \big(x'' - \frac{k+d-1}{d}(s_{c'} - s_{c''} - a'') - 1\big) \cdot \frac{d}{k+d-1} \\
			&= \frac{d}{t+d} - p^{+} + s_{c'} - s_{c^{+}} + x'' \cdot \frac{d}{k+d-1} - (s_{c'} - s_{c''} - a'') - \frac{d}{k+d-1} \\
			&= \frac{d}{t+d} - p^{+} + s_{c'} - s_{c^{+}} + x'' \cdot \frac{d}{k+d-1} - s_{c'} + s_{c''} + a'' - \frac{d}{k+d-1} \\
			&= \frac{d}{t+d} - p^{+} - s_{c^{+}} + x'' \cdot \frac{d}{k+d-1} + s_{c''} + a'' - \frac{d}{k+d-1} \\
			&\ge \frac{d}{t+d} - p^{+} - (s_{c''} + a'' - a^{+}) + x'' \cdot \frac{d}{k+d-1} + s_{c''} + a'' - \frac{d}{k+d-1} \\
			&= \frac{d}{t+d} - p^{+} - s_{c''} - a'' + a^{+} + x'' \cdot \frac{d}{k+d-1} + s_{c''} + a'' - \frac{d}{k+d-1} \\
			&= \frac{d}{t+d} - p^{+} + a^{+} + x'' \cdot \frac{d}{k+d-1} - \frac{d}{k+d-1} \\
			&\ge \frac{d}{t+d} - p^{+} + a^{+} - \frac{d}{k+d-1}
		\end{align*}
		
		We obtained the same expression, which we showed in the previous case to be non-negative, and we saw that the gap is strictly greater than it. Therefore, since the gap is strictly positive, we can state here as well that $W'$ will defeat $W^{**}$ under the truthful profile $(P_i, P_F, P_{\overline{F}})$.
		
		\medskip	
		
		\noindent\textbf{No $c'$-containing committee wins under the manipulated profile.}

		Now, after having proven that no committee that does not contain $c'$ will win under the profile $(P_i, P_F, P_{\overline{F}})$, we will prove that no committee that contains $c'$ will win under the profile $(A_i, P_F, P_{\overline{F}})$.
		
		If there is such a committee, we divide it into two cases: either it contains $c''$ or it does not contain $c''$.
		
		If it does not contain $c''$, we start from the committee $W'$, which we have already seen is defeated by $W''$ under the profile $(A_i, P_F, P_{\overline{F}})$, and we swap candidates one by one. Since the committee we are examining contains $c'$ and does not contain $c''$, any removal of a candidate from $W'$ will be of a candidate from the set $\{c_1, \ldots, c_{k-1}\}$, and any candidate addition will be an addition of a candidate from the set $C \setminus \{c_1, \ldots, c_{k-1}, c'', c'\}$. We have seen in Lemma \ref{lemma:swap_fact} that such a swap strictly decreases the score.
		
		If it does contain $c''$, we start from the committee $W^*$, which we have already seen is defeated by $W''$ under the profile $(A_i, P_F, P_{\overline{F}})$, and we swap candidates one by one. Since the committee we are examining contains $c'$ and also contains $c''$, any removal of a candidate from $W^*$ will be of a candidate from the set $\{c_1, \ldots, c_{k-1}\}$, and any candidate addition will be an addition of a candidate from the set $C \setminus \{c_1, \ldots, c_{k-1}, c'', c'\}$. We have seen in Lemma \ref{lemma:swap_fact} that such a swap strictly decreases the score.
		
		(Note that $W^*$ is not a single committee, but there could be several such committees depending on who was the first we removed from among $\{c_1, \ldots, c_{k-1}\}$, but the proof in \hyperlink{proof:drpav:step4}{Step 4} applied to each of them accordingly.)
		
		\bigskip
		
		Thus, we have seen that $c' \notin \mathcal{R}(A_i, P_F, P_{\overline{F}})$ and $c' \in \mathcal{R}(P_i, P_F, P_{\overline{F}})$. Since we defined that $c' \in P_i$, we have essentially found a case where $\mathcal{R}(A_i, P_F, P_{\overline{F}}) \cap P_i \supseteq \mathcal{R}(P_i, P_F, P_{\overline{F}}) \cap P_i$ does not hold, meaning the manipulation is not safe.
		
		\qed
	\end{proof}

\end{document}